\documentclass[11pt,letterpaper]{article}

\usepackage{geometry}
\usepackage{natbib}

\usepackage{titlesec}

\usepackage{microtype}
\usepackage{subfigure}
\usepackage{booktabs}
\usepackage{amsmath}
\usepackage{amsfonts}
\usepackage{amssymb}
\usepackage{amsthm}
\usepackage{amsbsy}
\usepackage{graphicx}
\usepackage{mathtools}

\usepackage{verbatim}
\usepackage{url}
\usepackage{bbm}
\usepackage{multirow}
\usepackage{enumitem}
\usepackage{xspace}
\usepackage{footnote}
\usepackage{hyperref}
\usepackage[svgnames]{xcolor}
\usepackage[capitalise,nameinlink]{cleveref}
\hypersetup{colorlinks={true},linkcolor={DarkBlue},citecolor=[named]{DarkGreen},urlcolor=black}

\usepackage{tikz}
\usepackage{algorithm}
\usepackage[noend]{algpseudocode}
\renewcommand{\algorithmiccomment}[1]{\bgroup\hfill\small\textcolor{gray}{//~#1}\egroup}
\algnotext{EndFor}
\algnotext{EndIf}
\algnotext{EndWhile}
\algnotext{EndProcedure}
\makeatletter
\renewcommand{\ALG@name}{Mechanism}
\makeatother

\usepackage{tikz}
\usepackage{pgfplots}
\pgfplotsset{compat=1.18}
\usetikzlibrary{math, calc}

\usepackage{bm}

\usepackage{titling}
\usepackage[noblocks]{authblk}

\theoremstyle{plain}
\newtheorem{theorem}{Theorem}[section]

\newtheorem{lemma}[theorem]{Lemma}
\newtheorem{corollary}[theorem]{Corollary}
\theoremstyle{definition}
\newtheorem{definition}[theorem]{Definition}

\newtheorem{example}[theorem]{Example}

\usepackage[tt=false]{libertine}

\usepackage{soul}
\setuldepth{p}

\newcommand{\mms}{\ensuremath{\textrm{\MakeUppercase{mms}}}\xspace}

\newcommand{\efo}{\ensuremath{\textrm{\MakeUppercase{ef}1}}\xspace}

\newcommand{\opt}{\ensuremath{\mathrm{OPT}}}
\newcommand{\optm}{\ensuremath{\mathrm{OPT}^-}}
\newcommand{\fmms}{\ensuremath{\mathrm{FMMS}}}

\DeclareMathOperator*{\argmax}{arg\,max}
\DeclareMathOperator{\E}{\mathbb{E}}

\newcommand{\multisub}{\ensuremath{\textsc{MultiSubmod}}}

\allowdisplaybreaks

\usepackage[skip=.4\baselineskip plus 2.5pt minus 1.0pt]{parskip}
\begin{document}

\author[1,2]{Georgios Amanatidis}
\affil[1]{Department of Informatics, Athens University of Economics and Business, Greece.}
\affil[2]{Archimedes,  Athena Research Center, Greece.}
\author[3]{Georgios Birmpas}
\affil[3]{Department of Computer Science, University of Liverpool, UK.}
\author[4]{Philip Lazos}
\affil[4]{Independent Researcher. }
\author[5]{\\ Stefano Leonardi}
\affil[5]{Department of Computer, Control, and Management Engineering, Sapienza University of Rome, Italy. }
\author[6]{Rebecca Reiffenh{\"a}user}
\affil[6]{Institute for Logic, Language and Computation, University of Amsterdam, The Netherlands.  }

\title{Algorithmically Fair Maximization of Multiple Submodular Objective Functions and Implications to Constrained Fair Division}

\predate{}
\postdate{}
\date{}

\maketitle

\begin{abstract}
Constrained maximization of submodular functions is a central problem in combinatorial optimization. In many realistic scenarios, multiple agents each need to maximize their own submodular objective over a common ground set, subject to individual combinatorial constraints, with the requirement that their solutions be disjoint. We study this setting through the lens of algorithmic fairness and constrained fair division. Inspired by the fair division literature, we propose and analyze a simple Round-Robin protocol in which agents take turns building their solutions one item at a time; each agent is free to use any internal algorithm, and the protocol itself performs no computation. We show that agents following simple greedy policies enjoy solid guarantees for both monotone and non-monotone objectives subject to combinatorial constraints as general as $p$-systems, which capture cardinality and matroid intersection constraints. For monotone objectives, a greedy agent $i$ with a $p_i$-system constraint achieves a $1/(n+p_i)$ fraction of the best value available when they first get to choose. On instances that are robust to competition---where no single agent's optimal value is greatly affected by losing a bounded number of items to others---these guarantees improve to a $1/\Theta(p_i)$ approximation of the unconstrained optimum, which is asymptotically best-possible in polynomial time. We further establish novel fairness guarantees: greedy agents produce approximately feasible-envy-free-up-to-one-item (FEF1) and approximately feasible-envy-free-towards-unallocated-items (FEFu) allocations for both monotone and non-monotone objectives. Via a simple augmented protocol and a self-contained polynomial-time proxy algorithm, we also obtain the first $\Theta(1/p_i)$-approximate feasible maximin share (FMMS) guarantees for submodular agents with general combinatorial constraints. Finally, although greedy policies may not be individually optimal, we show that consistently improving upon them is NP-hard even in the simplest settings.
\end{abstract}

\section{Introduction}
\label{sec:intro}
Dealing with competing interests of selfish agents poses the main challenge at the heart of many research directions.
In \textit{auction design}, one aims to elicit honest bidding in pursuit of a welfare-maximizing assignment of items; in \textit{social choice}, the goal is to map individual preferences to outcomes that serve society as a whole; in \textit{fair division}, one attempts to serve agents' interests in a way that satisfies some notion of equal treatment.

We study a natural setting that shares a similar flavor: there are $n$ agents, each of whom wants to find a subset of discrete resources that maximizes their own submodular objective function, subject to a combinatorial constraint (such as a cardinality or matroid intersection constraint).
The agents' solutions must be disjoint, so they inevitably compete for resources.
What distinguishes our approach from classic algorithmic design is that we do \emph{not} want to solve a centralized combinatorial optimization problem on behalf of the agents.
Instead, we ask: can a simple, instance-independent \emph{protocol}---one that the agents themselves agree is algorithmically fair---allow each agent to obtain a provably good solution?

As an illustrating example, consider the problem of influence maximization, e.g., as defined in \citet{breuer2020fast}.
Multiple companies want to advertise on the same social network (e.g., Facebook) by selecting disjoint sets of seed users.
Each company solves its own influence maximization problem.
How should one implement a selection process that the companies would agree is fair?
And how should one exploit the fact that, at large scale, an agent's optimal value is not greatly affected by the removal of a few contested items?

We address these questions by investigating the Round-Robin protocol, where agents take turns adding one item at a time to their respective solutions.
Inspired by the fair division literature, this protocol is transparent, non-adaptive (its order does not depend on the agents' valuations or constraints), and delegates all computational work to the agents themselves.
Our main finding is that greedy policies---which are both simple and polynomial-time---yield strong guarantees in this protocol, both from an optimization standpoint and from a fairness standpoint.

From an algorithm or mechanism design point of view, contrary to classic approaches in social choice or auctions, we never collect the valuations, nor are we interested in eliciting truthful reports. 
Instead we opt for the much simpler solution of defining a selection protocol and then show that it is possible for every agent to have a guarantee on the quality of their final solution, even by following a simple policy. 
Note that we do not imply any specific way this should be done; whether the agents indeed realize their potential, or what algorithms are actually used in the background, is of no direct interest. Rather, the aim here is to provide the participants with an agreeable framework that reasonably limits the negative effects of competition and avoids a possible \textit{winner-takes-all} situation.

From an optimization perspective, a clear advantage of our approach is that it allows every agent to get a reasonably good solution.
Assuming there are $n$ agents, each with an objective function $f_i$ over a set of items $M$, note that no algorithmic framework for the problem can guarantee an approximation factor of more than $1/n$ to all the agents with respect to the optimal value they could obtain without competition, even when randomization is allowed. This is straightforward to see in the example where $f_i(S) = |S|$, for all $S\subseteq M$ and $i\in [n]$. The best worst-case guarantee for the least satisfied agent is $1/n$ 
and, thus, we would like an algorithmic framework with comparable guarantees for highly competitive instances like this,  which ideally can significantly improve for ``nicer'' instances.
What we show is that by following a simple greedy policy throughout our Round-Robin protocol, an agent can achieve such guarantees, even when they have strong feasibility constraints. While we are not interested in enforcing any specific policy, we show that improving over these greedy policies is NP-hard for the agents, even in the very simplest of settings. 
Interestingly enough, as one moves towards instances that are more robust to competition---under a \textit{large market} flavored assumption that we call $(\alpha, \beta)$-robustness (Definition \ref{def:robust})---the guarantees of the protocol improve up to the point where constant approximation guarantees can be achieved for everyone. Such instances, where losing an item to another agent has a limited effect, are closer to what is often seen in real-world applications like the aforementioned influence maximization example.

Finally, there is an obvious connection of our problem with fair division, even in the choice of the protocol itself. If the items in $M$ were allocated in a centralized way, one could directly interpret our problem as a fair division problem where the items are indivisible and the agents have submodular valuation functions as well as feasibility constraints over the subsets of goods. Indeed, we explore these connections in depth. On one hand, in Sections \ref{sec:fairness-monotone} and \ref{sec:fairness-non-monotone} we show that our optimization results imply \textit{feasible envy-freeness up to one item} (FEF1) results along with approximate \textit{feasible envy-freeness towards the unallocated items} (FEFu). On the other hand, by interpreting $(\alpha, \beta)$-robustness against the maximin share of an agent instead of their optimal value, in Section \ref{sec:mms} we modify our Round-Robin protocol to allow for maximin share guarantees. To the best of our knowledge, both groups of results are the first of their kind for instances with so rich combinatorial constraints.

\smallskip

\noindent{\textbf{Our Contribution.}}
We initiate the study of \emph{coordinated maximization protocols} for the problem where $n$ agents want to maximize a submodular function each, say $f_i$ for agent $i$, subject to a combinatorial constraint; all the functions are defined over a common ground set $M$ and the agents' solutions should be disjoint.
We suggest and analyze a natural Round-Robin protocol (Protocol~\ref{alg:MRR}) and a randomized variant (Protocol~\ref{alg:RMRR}), as well as an Augmented Round-Robin protocol (Protocol~\ref{alg:AugRR}) designed for maximin share guarantees.
Our main results are:

\begin{itemize}

\item \textbf{Optimization guarantees via greedy policies.}
An agent $i$ with a monotone submodular objective and a $p_i$-system constraint who chooses greedily achieves value at least $\optm_i/(n+p_i)$ (Theorem~\ref{thm:greedy_monotone_general}), improving to $\optm_i/n$ for cardinality constraints (Theorem~\ref{thm:greedy_monotone_cardinality}).
For non-monotone objectives, two simultaneous greedy solutions yield value at least $\optm_i/(4n+4p_i+2)$ (Theorem~\ref{thm:greedy_non-monotone_general}), improving to $\optm_i/(4n+2)$ for cardinality constraints (Theorem~\ref{thm:greedy_non-monotone_cardinality}).

\item \textbf{FEF1 and FEFu fairness guarantees.}
A greedy agent $i$ with a monotone objective and a $p_i$-system constraint achieves a $1/(p_i+2)$-FEF1 guarantee with respect to any other agent's bundle (Theorem~\ref{thm:greedy_monotone_ef1}), and a $1/(p_i+1)$-FEFu guarantee with respect to unallocated items.
For cardinality constraints this improves to $0.5$-FEF1 and $0.5$-FEFu (Theorem~\ref{thm:greedy_monotone_cardinality_ef1}).
Simulating Round-Robin with greedy policies for all agents efficiently produces a $\frac{1}{p+2}$-FEF1 maximal allocation (Corollary~\ref{cor:fairness_p-system}).
Analogous, though weaker, results hold for non-monotone objectives (Section~\ref{sec:fairness-non-monotone}).

\item \textbf{Maximin share guarantees.}
Via the Augmented Round-Robin protocol and a self-contained polynomial-time proxy (Algorithm~\ref{alg:maximin_proxy}), each agent $i$ with a monotone objective and a $p_i$-system (resp.\ cardinality) constraint achieves value at least $\fmms_i/(2p_i+5)$ (resp.\ $\fmms_i/5$), where $\fmms_i$ is agent $i$'s feasible maximin share (Theorem~\ref{thm:greedy_monotone_general_MMS}).
The first $\Theta(1/p_i)$-approximate FMMS guarantees for non-monotone objectives are also established (Theorem~\ref{thm:greedy_non-monotone_general_MMS}).

\item \textbf{Improved guarantees on robust instances.}
On instances that are $(\Omega(n),O(1))$-robust with respect to every agent, greedy policies achieve $1/\Theta(p_i)$ of $\opt_i$ (Corollaries~\ref{cor:greedy_monotone_general_robust_2} and~\ref{cor:greedy_non-monotone_general_robust_2}), which is asymptotically best-possible in polynomial time.
Randomizing over the agent ordering (Protocol~\ref{alg:RMRR}) yields comparable ex-ante guarantees with respect to $\opt_i$ even on non-robust instances (Theorem~\ref{thm:randomized}).
The robustness-based analysis is also the key ingredient for the maximin share results above.

\item \textbf{Hardness of improvement.}
Even when all other agents have additive objectives and follow greedy policies, consistently achieving a $(1+\varepsilon)$-improvement over the greedy policy is NP-hard for any small $\varepsilon$ (Theorem~\ref{thm:second_hardness}).
\end{itemize}

The analysis of greedy policies in the Round-Robin protocol requires novel combinatorial mappings that associate items added to an agent's solution with items ``lost'' due to myopic choices or competition, simultaneously for multiple disjoint reference sets (Lemmata~\ref{lem:mapping-monotone-general} and~\ref{lem:mapping-non-monotone-general}).
These mappings are the technical foundation for all our results.

\smallskip

\noindent{\textbf{Technical Challenges.}}
The central technical difficulty is that in the Round-Robin protocol an agent may fail to include an item in their solution for two fundamentally different reasons: either their own greedy choices have rendered it infeasible (because its addition would violate the independence constraint), or another agent has already taken it. Existing analyses of greedy algorithms for constrained submodular maximization handle only the former, and naively combining both sources of loss would yield an approximation factor of order $\Theta(np_i)$ rather than the $\Theta(n + p_i)$ we obtain. The key technical result that overcomes this is a novel item-mapping lemma (Lemma~\ref{lem:mapping-monotone-general} and its counterpart for non-monotone objectives, Lemma~\ref{lem:mapping-non-monotone-general}). At a high level, we partition the items an agent fails to collect into two groups: those lost to other agents between two consecutive choices of agent~$i$ (at most $n - 1$ per round, each with marginal value no better than the next greedy pick), and those that became infeasible due to the agent's own previous choices. Bounding the latter group requires a careful inductive argument using the $p$-system definition: after $r$ greedy picks, at most $rp_i$ items could have been rendered infeasible, allowing us to map at most $p_i$ of them to each greedy choice of agent~$i$ without violating the marginal-value ordering. The lemma is stated for $a$ disjoint feasible reference sets (not just the single optimal solution), a generalization that may seem purely technical at first but turns out to be the key to both the robust-instance results and the maximin share analysis, where $a = \Theta(n)$ reference sets arise naturally from the  definitions of  robustness and of the feasible maximin share.

The maximin share results (Section~\ref{sec:mms}) rest on a further structural observation: if no available singleton is worth at least $\fmms_i/\theta$, then the remaining instance is automatically $\Omega(n)$-robust (a maximin share defining partition provides $n$ disjoint feasible sets of significant value), so the robust-instance machinery directly applies. When $\fmms_i$ is unknown---computing it is a hard problem in general---we show that an agent can efficiently compute a proxy value in $[\fmms_i/\theta, \fmms_i]$ by simulating $O(n)$ executions of the greedy Round-Robin protocol on copies of themselves, incurring only a constant-factor loss in the final guarantee. 

\smallskip

\noindent{\textbf{Related Work.}}
There is a vast literature on optimizing a submodular function, with or without constraints, dating back to the seminal works of \citet{NemhauserWF78,FNemhauserW78}. For an overview of the main variants of the problem and the corresponding results, we refer the interested reader to the recent survey of \citet{BuchbinderF18} and the references therein. 
For maximizing a monotone submodular function subject to a $p$-system constraint, in particular, the simple greedy algorithm that adds to the solution the item with the highest marginal value in each step achieves an approximation factor of $1/(p+1)$ \citep{FNemhauserW78}, which improves to $1-1/e$ for the case of a cardinality constraint \citep{NemhauserWF78}. The latter is the best possible approximation factor, assuming that $\mathrm{P} \neq \mathrm{NP}$ \citep[see Theorem \ref{thm:Feige-hardness}]{Feige98}. 
Moreover, \citet{BadanidiyuruV14} showed that even for maximizing an \textit{additive} function subject to a $p$-system constraint, achieving an approximation factor of $1/(p-\varepsilon)$ for any fixed $\varepsilon > 0$, requires
exponentially many independence oracle queries.
The state-of-the-art for non-monotone objective functions subject to a $p$-system constraint is much more recent.
\citet{GuptaRST10} introduced a \textit{repeated greedy} framework, which can achieve an approximation factor of $1/2p$, as shown by  \citet{MirzasoleimanBK16}. The best known factor is $1/(p+\sqrt{p\,})$ by \citet{FeldmanHK23}, using the \textit{simultaneous greedy} framework. Simultaneous greedy algorithms bypass non-monotonicity by constructing \emph{multiple} greedy solutions \emph{at the same time}. This idea was first introduced by \citet{AmanatidisKS22} for a knapsack constraint and a $p$-system constraint in a mechanism design setting, and has been used successfully in a number of variants of the problem since \citep{Kuhnle19,HanCCW20,SunZZZ22,PhamT0T23,ChenK23}. 
The impressive versatility of this approach, however, became apparent in the recent work of \citet{FeldmanHK23}, who fully develop a framework that achieves the best known factors for a $p$-system constraint and multiple knapsack constraints combined.

Fair division, both with divisible and with indivisible resources, has been an area of growing popularity over the last years. For a recent survey on fair division of indivisible items, see \citet{AmanatidisABFLMVW23}. The Round-Robin algorithm, which underlies our protocols, is a fundamental
procedure encountered throughout the fair division literature \citep{GW17,CaragiannisKMPS19,ManurangsiS21}, often modified \citep[e.g.,][]{AMNS17,BarmanK20,AzizCIW22} or as a subroutine of more complex algorithms \citep[e.g.,][]{Kurokawa17,ANM2019}. 
The problem of fairly allocating goods subject to cardinality, matroid, or even more general, constraints has been studied before; see \citet{Suksompong21} for a recent survey. 
Nevertheless, to the best of our knowledge, these works only focus on additive valuation functions (or their restrictions with respect to the constraints) and do not study constraints as general as $p$-systems, with the exception of \citet{LiV21} who, however, assume that the same constraint is common to all the agents.
What is probably more relevant to this work is the rich line of work on discrete fair division beyond additive valuation functions, e.g., \citep{BarmanK20,GhodsiHSSY22,AmanatidisBL0R23,GargKK23,UziahuF23,ChrisCMS25,ChekuriKKM24}. In particular, \citet{BarmanK20} showed the existence and polynomial
time computability of $0.21$-MMS allocations for agents with submodular valuation functions. 
This approximation factor was later improved to $1/3$ by \cite{GhodsiHSSY22} and to $10/27$ by \citet{UziahuF23}. The only work that systematically studies non-monotone, non-additive valuation functions is the very recent work of \citet{BarmanV26}. Although in that work the authors study much more general functions than non-monotone submodular, this scarcity of fair division results for non-monotone objectives is indicative of the challenging nature of the problem.

\smallskip

\noindent{\textbf{Comparison to the Conference Version.}} Compared to its conference version \citep{ABLLR25conf}, this work has been thoroughly revised and significantly expanded. The main new contributions are the fairness results: the FEF1 and FEFu guarantees for non-monotone objectives (Section~\ref{sec:fairness-non-monotone}) and the feasible maximin share guarantees via the Augmented Round-Robin protocol and the Maximin Share Proxy algorithm (Section~\ref{sec:mms}) are entirely new. To accommodate the stronger results in Section~\ref{sec:mms}, the definition of $(\alpha,\beta)$-robustness has been revised to involve a sum over $\alpha$ feasible sets rather than a worst-case bound, and the analyses in Sections~\ref{sec:robust_monotone} and~\ref{sec:robust_non-monotone} have been updated accordingly.

\section{Preliminaries}
\label{sec:prelims}

Here we introduce some notation and give the relevant definitions and some basic facts.
Let  $M = [m] = \{1, 2, ..., m\}$ be a set of $m$ items. 
For a function $f: 2^M \rightarrow \mathbb{R}$ and any sets $S\subseteq T \subseteq M$ we use the shortcut $f(T\,|\,S)$ for 
the \emph{marginal value of\,\ $T$ with respect to $S$}, i.e., $f(T\,|\,S)=f(T \cup  S) - f(S)$. If $T = \{i\}$ we simply write $f(i\,|\,S)$.

\begin{definition}
The function $f$ 
is 	\emph{submodular} if and only if $f(i\,|\,S) \geq f(i\,|\,T)$ for all $S\subseteq T \subseteq M$ and $i\not\in T$.   
\end{definition}

If $f$ is non-decreasing, i.e., if $f(S) \le f(T)$ for any $S \subseteq T \subseteq M$, we just refer to it as being \emph{monotone} in this context. We 
consider normalized (i.e., 
$f(\emptyset)=0$), non-negative submodular objective functions, both monotone (Sections \ref{sec:positive_monotone}, \ref{sec:hardness_monotone}, \ref{sec:randomizedRR}) and non-monotone (Sections \ref{sec:non-monotone}, \ref{sec:randomizedRR}).

Our algorithmic goal is to maximize multiple submodular functions at the same time, each one with its own constraint, through simple protocols.  As we imply in the Introduction,  the term \textit{protocol} here refers to a procedure that (a) does not take any general input about the valuation functions or the constraints of the agents but is, possibly, allowed to ask a limited number of simple queries, (b) performs little to no computation itself, and (c) chooses one agent at a time and allows them to add a number of items to their solution, possibly from a subset of the available items.

The constraints we consider can be as general as $p$-system constraints. 

\begin{definition}\label{def:p-systems}
Given a set $M$, an \emph{independence system} for $M$ is a family $\mathcal{I}$ of subsets of $M$, whose members are called the \emph{independent sets} of $M$ and satisfy \emph{(i)} $\emptyset \in \mathcal{I}$, and \emph{(ii)} if $B \in \mathcal{I}$ and $A \subseteq B$, then $A \in \mathcal{I}$.
We call $\mathcal{I}$ a \emph{matroid} if it is an independence system and it also satisfies the exchange property \emph{(iii)} if $A, B \in \mathcal{I}$ and $|A| < |B|$,
then there exists $x \in  B\setminus A$ such that $A \cup \{x\} \in \mathcal{I}$. 
\end{definition}

Given a set $S \subseteq M$, a maximal independent set contained in $S$ is called a \emph{basis} of $S$. The upper rank $\textrm{ur}(S)$ (resp.~lower rank $\textrm{lr}(S)$) is defined as the largest (resp.~smallest) cardinality of a basis of $S$. 

\begin{definition}
A \emph{$p$-system} for $M$ is an independence system for $M$, such that $\displaystyle \max_{S\subseteq M} \textrm{ur}(S)/\textrm{lr}(S) \le p$.    
\end{definition}

Many combinatorial constraints are special cases of $p$-systems for small values of $p$. A \emph{cardinality} constraint, i.e., feasible solutions contain up to a certain number of items, induces a $1$-system. A \emph{matroid} constraint, i.e., feasible solutions belong to a given matroid, also induces a 1-system; in fact, a cardinality constraint is a special case of a matroid constraint.  More generally, constraints imposed by the intersection of $k$ matroids, i.e., feasible solutions belong to the intersection of $k$ given matroids, induce a $k$-system; matching constraints are examples of such constraints for $k=2$.
\vspace{0pt}

\noindent\emph{Constrained Maximization of Multiple Submodular Objectives} ($\multisub$ for short): Let $f_i : 2^M \rightarrow \mathbb{R}$, $i\in[n]$, be a submodular function and $\mathcal{I}_i\subseteq 2^M$, $i\in[n]$, be a $p_i$-system. Find disjoint subsets of $M$, say $S_1, \ldots, S_n$, such that $S_i\in \mathcal{I}_i$ and $f_i(S_i) = \max_{S\in \mathcal{I}_i} f_i(S) \vcentcolon= \opt_i$.
\vspace{4pt}

We think of $f_i$ and $\mathcal{I}_i$ as being associated with an agent $i\in[n]$, i.e., they are $i$'s objective function and combinatorial constraint, respectively. Of course, maximizing all the functions at once may be impossible as  these objectives could be  competing with each other. Naturally, we aim for approximate solutions, i.e., for $S_1, \ldots, S_n$, such that $f_i(S_i) \ge \rho\, \opt_i$, for all $i\in[n]$ and a common approximation ratio $\rho$ (possibly a function of $n$). As a necessary compromise in our setting, we often use $\optm_i$ instead of $\opt_i$ as the benchmark for agent $i$ in the worst-case.
We will revisit and formalize this benchmark in Section \ref{sec:benchmarks} but, essentially, if some items have already been allocated right before anything is added to $S_i$, then $\optm_i$ is 
the value of an 
optimal solution for $i$ still remaining available at that time.

As we mentioned in the Introduction, we also intend to  evaluate our protocols on instances that are more robust to competition. We want to capture the behavior one would expect in large-scale applications of our setting, e.g., in our running example of multiple firms competing to maximize their influence on a vast social network. That is, the value of an optimal solution of an agent should not be greatly affected by the removal of a reasonably sized subset of items. The next definition formalizes this idea. 

\begin{definition}\label{def:robust}
    Let $\alpha \in \mathbb{N}, \beta \in \mathbb{R}_+$. An instance of $\multisub$ is \emph{$(\alpha, \beta)$-robust with respect to agent $i$} if there are $\alpha$ disjoint subsets of $M$, say $O_{i1}, \ldots, O_{i \alpha}$, such that $O_{ij}\in \mathcal{I}_i$ and \[\frac{1}{\alpha} \cdot \sum_{j\in [\alpha]}f_i(O_{ij}) \ge \frac{\opt_i}{\beta}\,. \] 
\end{definition}
That is, if an instance is $(\alpha, \beta)$-robust with respect to agent $i$, then it contains at least $\alpha$ independent solutions of average value within a factor of $\beta$ from $i$'s optimal value. Clearly, any instance is $(1, 1)$-robust with respect to any agent. When we refer to instances that are more robust to competition, we essentially mean instances that are $(\Omega(n), O(1))$-robust with respect to everyone.

Besides the definition of submodularity given above, there are alternative equivalent definitions that will be useful later. These are summarized in the following result of \citet{NemhauserWF78}.

\begin{theorem}[\citet{NemhauserWF78}]\label{thm:SM-general}\label{thm:SM-monotone}
	A  function $f:2^M \rightarrow \mathbb{R}$ is submodular if and only if, for all $S, T \subseteq M$, 
  \[f(T) \le f(S) + \sum_{i \in T \setminus S} f(i\,|\,S) - \sum_{i \in S 
		\setminus T} f(i\,|\,S\cup T \setminus \{i\}) \,.\]
Further, $f$ is monotone submodular if and only if, for all $S, T \subseteq M$, 
\[f(T) \le f(S) + \sum_{i \in T \setminus S} f(i\,|\,S)\,.\]
\end{theorem}

As it is common in the submodular optimization literature, we assume oracle access to the functions via value queries, i.e., for  $i\in[n]$, we assume the existence of a polynomial-time value oracle that returns $f_i(S)$ when given as input a set $S$. 
Similarly, we assume the existence of independence oracles for the constraints, i.e., for  $i\in[n]$, we assume there is a polynomial-time algorithm that, given as input a set $S$, decides whether $S\in \mathcal{I}_i$ or not.

\subsection{Our Round-Robin Framework} 
\label{sec:round-robin}
We present our simple protocol, which is aligned with how the majority of submodular maximization algorithms work (i.e., building one or more solutions one item at a time): here, the agents take turns according to a fixed ordering and in each step the active agent chooses (at most) one available item to add to their solution. 
Note that we do not impose how this should be done; it is the agents' task to decide how an item will be chosen, whether their solution should remain feasible throughout the protocol or they maintain a feasible solution within a larger chosen set, etc.  
We stress again that this approach has the significant advantage of delegating any computationally challenging task to the agents.

Note that when we refer to the \textit{policy} of an agent, we mean their overall algorithmic strategy; how they make their algorithmic choices, given full information about other agents' objective functions, constraints, and current solutions. So when we write $\mathcal{A}_i(S_i\,; Q)$ in line \ref{line:rr4} of the description of Protocol \ref{alg:MRR}, in general, $\mathcal{A}_i(\cdot)$ can be a function of all that information. Later, in Section \ref{sec:hardness_monotone}, when we make the distinction and talk about the \textit{algorithm} of an agent, we typically consider other agents' objective functions and policies fixed. We do not formalize this further, as the main policies we consider in this work are independent of any information about other agents.

\makeatletter
\renewcommand{\ALG@name}{Protocol}
\makeatother

\begin{algorithm}[H]
		\caption{Round-Robin$({\mathcal{A}}_1, \ldots, {\mathcal{A}}_n)$ \\{\small {(For $i\in [n]$, ${\mathcal{A}}_i$ is the policy of   agent $i$.)}}}
		\begin{algorithmic}[1]
			\State $Q=M$\textbf{;} $k = \lceil m/n\rceil$ 
			\For{$r = 1, \dots, k$} 
			\For{$i = 1, \dots, n$} 
			\State $j = \mathcal{A}_i(S_i\,; Q)$ {\small\hfill (where $j$ could be a \textit{dummy} item)} \label{line:rr4}
			\State $Q = Q\setminus \{j\}$ \label{line:rr6}
			\EndFor
			\EndFor
		\end{algorithmic}
		\label{alg:MRR}
\end{algorithm}

We could have any ordering fixed by Protocol \ref{alg:MRR} at its very beginning, i.e., a permutation $s_1, \ldots, s_n$  of $[n]$, such that $s_i$ is the $i$-th agent to choose their first item. To simplify the presentation, we assume that $s_i = i$, for all $i\in[n]$. This is without loss of generality, as it only involves a renaming of the agents before the main part of the protocol begins. We revisit this convention in Section \ref{sec:randomizedRR} where we randomize over all possible agent permutations.

Also, in Section \ref{sec:mms} we augment Protocol \ref{alg:MRR} by adding a simple initial phase that allows us to argue about maximin share guarantees; see Protocol \ref{alg:AugRR}.

\subsection{A Refined Benchmark and Fairness Notions}
\label{sec:benchmarks}
Note that even having a guarantee of $\opt_i / n$, for all $i\in [n]$, is not always possible (even approximately) in the worst case. Indeed, let us modify the example from the introduction so that there are a few very valuable items: for all $i\in [n]$, let 
$f_i(S\cup T) = |S| + L|T|$, 
for any $S\subseteq M_1$ and $T\subseteq M_2$, where $M= M_1\cup M_2$ and $L\gg |M|$. 
When $|M_2| < n$, 
then the best possible value for the least satisfied agent is a $1/(n-|M_2|)$ fraction \textit{not of\, $\opt_i$ but of $i$'s optimal value in a reduced instance where the items in $M_2$ are already gone}.

With this in mind, we are going to relax the $\opt_i$ benchmark a little. As mentioned in Section \ref{sec:prelims}, we define $\optm_i$ to be 
the value of an
optimal solution available to agent $i$, given that $i-1$ items have been lost before $i$ gets to pick their first item. 
Note that, for $i\in[n]$,
\begin{equation}\label{eq:optm_pessimistic}
\optm_i \ge \min_{M'\in \binom{M}{m-i+1}}\max_{S\in \mathcal{I}_i |M'}f_i(S) \,,
\end{equation}
i.e., $\optm_i$ is always at least as large as the pessimistic prediction that agents before $i$ will make the worst possible choices for $i$. 
The notation $\binom{M}{x}$ used here denotes the set of subsets of $M$ of cardinality $x$ and $\mathcal{I}|A$ denotes (the independence system induced by) the restriction of $\mathcal{I}$ on $A$, i.e., $\mathcal{I}|A = \{X\cap A:\, X\in \mathcal{I}\}$.
By inspecting Definition \ref{def:p-systems},
it is easy to see that the restriction of a $p$-system is a $p$-system.  

It is worth mentioning that, while $\optm_i$ is defined with respect to Protocol \ref{alg:MRR} here, its essence \textit{is not an artifact} of how Round-Robin works.  No matter which sequential protocol (or even algorithm with full information) one uses, if there are $n$ agents, then it is unavoidable that someone will get their first item after $r$ items are already gone, for any $r<n$; this is inherent to any setting with agents competing for resources.

\noindent\textbf{Maximin share fairness.} Still, in instances like the aforementioned example, $\optm_i$ might be a fairly weak guarantee. Indeed, from agent $i$'s perspective, the only fair outcome would be that they receive at least one item from $M_2$ or the entire $M_1$. A possible way to quantify this is to resort to the notion of a \emph{maximin share} (MMS) from the fair division literature, which was introduced by \citet{Budish11} and was popularized by \cite{ProcacciaW14} and subsequent works. Intuitively, the maximin share of an agent is the value they could guarantee to themselves by partitioning all the items in $n$ subsets and taking their least favorite among those. Of course, for our setting, one would have to adjust this thought experiment in order for the resulting notion to be meaningful; an agent should imagine that they create $n$ disjoint subsets that are \emph{feasible} for them and getting their least favorite among them. 

\begin{definition}\label{def:mmshare}
	Given  a subset  $M'\subseteq M$, the $n$-\textit{feasible maximin share} of agent $i$ with respect to $M'$ is
	\[ \fmms_i(n, M') = \displaystyle\max_{\mathcal{T}\in\mathcal{F}_n(M')}\ \min_{T_j\in \mathcal{T}} f_i(T_j)\,,\]
	where $\mathcal{F}_n(M') = \{(S_1, \ldots, S_n) \,:\, S_j\in \mathcal{I}_i, \text{ for all $j$, and } S_j\cap S_k = \emptyset, \text{ for all distinct $j$ and $k$}\}$  is the set of all $n$-tuples of disjoint subsets of $M'$ which are feasible for agent $i$. We say that $\mathcal{T} \in \mathcal{F}_n(M')$ is an $(n, M')$-\textit{maximin share defining tuple} for agent $i$, if $\min_{T_j\in \mathcal{T}} f_i(T_j) = \fmms_i(n, M')$.
\end{definition}

When $M'=M$, we call $\fmms_i(n, M)$ the \textit{feasible maximin share} of agent $i$ and simply write $\fmms_i$. 

One may note that the right-hand side of \eqref{eq:optm_pessimistic} has a similar flavor to the mental experiment behind the definition of \fmms.  However, $\optm_i$ and $\fmms_i$ are not related in a straightforward way. It is easy to see that $\fmms_i \le \optm_i$ (and occasionally this holds with equality) but $\optm_i$ can be \textit{arbitrarily larger} than  $\fmms_i$. Even among $\optm_i / n$ and $\fmms_i$ (that we essentially use as our benchmarks), 
neither one is always a stronger guarantee than the other.
In the above example for $|M_1| = n$ and $|M_2| = n-1$, we have $\fmms_i = |M_1| = n  = \optm_i = n\cdot \optm_i / n$, but for $|M_1| = n-i-1$ and $|M_2| = i$, we have $\fmms_i = 0$ and $\optm_i / n = (L +n -i-1)/n$, for $i\le n-1$.

In Section \ref{sec:mms}, we show that when agent $i$ has a $p_i$-system constraint, following a simple greedy policy in Protocol \ref{alg:AugRR} (a modification of Protocol \ref{alg:MRR}) can guarantee them value of at least $\fmms_i / \Theta(p_i)$. Toward this, and our other \fmms-related results, we will need the following monotonicity property which generalizes (the $k=1$ case of) Lemma 1  of  \citet{BouveretL16}.

\begin{lemma}\label{lem:monotonicity}
	For any $n\ge 3$, any subset $M'\subseteq M$ and any $x \in M'$, it holds that 
	\[\fmms_i(n-1, M'\setminus \{x\}) \geq \fmms_i(n, M')\,.\]
\end{lemma}

\begin{proof}
	Consider an $(n, M')$-maximin share defining tuple for agent $i$, say $\mathcal{T}= (T_1, \ldots,T_n)$. Now, $x$ may or may not belong to one of the sets of $\mathcal{T}$; if it does, we may assume, without loss of generality, that $x \in T_n$. 
	In either case, consider the remaining partition $\mathcal{T}'=(T_1,\ldots,T_{n-1})$. Clearly, $\mathcal{T}' \in \mathcal{F}_{n-1}(M'\setminus \{x\})$ and, moreover, $f_i(T_j) \ge \fmms_i(n, M')$ for any $T_j \in \mathcal{T}'$. Thus,
	we have $\fmms_i(n-1, M'\setminus \{x\}) \geq \fmms_i(n, M')$.
\end{proof}

\noindent\textbf{Envy-freeness up to an item.} Next we focus  on  appropriate generalizations of  \textit{envy-freeness up to an item} (EF1)  and of \textit{maximin share fairness} (MMS). Both the original notions were introduced by \citet{Budish11} (and, implicitly, by \citet{LMMS04} a few years earlier in the case of EF1).  
An \emph{allocation} is a tuple of disjoint subsets of $M$, $\mathcal{A}= \allowbreak (A_1,\ldots,A_n)$, such that each agent $i \in [n]$ receives the set $A_i$. Note that here we do not assume the allocation to be complete, because of the presence of the feasibility constraints in our setting; that is, $\bigcup_{i \in [n]}A_i$ may be a strict subset of $M$.

\begin{definition} \label{def:ef1}
An allocation $\mathcal{A}$ is {\em $\rho$-approximate envy-free up to one item ($\rho$-EF1)} if, for every pair of agents $i,j \in [n]$, either $A_j = \emptyset$, or there is some $g \in A_j$, such that $f_i(A_i) \geq \rho\, f_i(A_j \setminus \{g\})$.
\end{definition}

When $\rho = 1$, we just refer to EF1 allocations. Round-Robin, if implemented as an algorithm with all agents choosing greedily, is known to produce EF1 allocations when agents have additive valuation functions \citep{CaragiannisKMPS19} and $0.5$-EF1 allocations when  agents have submodular valuation functions \citep{AmanatidisBL0R23}, assuming no constraints at all.
Although it is possible to have EF1 allocations under cardinality constraints for agents with additive valuation functions \citep{BiswasB18}, this definition that ignores feasibility is way too strong for our general constraints. What we need here is the notion of \textit{feasible EF1} introduced recently by \citet{DrorFS23} and others (see, e.g., \citet{Barman0SS23}), or rather its approximate version.

\begin{definition} \label{def:f-ef1}
An allocation $\mathcal{A}$ is {\em $\rho$-approximate feasible envy-free up to one item ($\rho$-FEF1)} if, for every pair of agents $i,j \in [n]$, either $A_j = \emptyset$, or  there is some $g \in A_j$, such that $f_i(A_i) \geq \rho\, f_i(A'_j)$ for any $A'_j \subseteq A_j \setminus \{g\}$ that is feasible for agent $i$.
\end{definition}

Since we are dealing with allocations that may not be complete here, in order for the notion of FEF1 to be meaningful, we also need the agents to have a similar guarantee with respect to the set of unallocated items.

\begin{definition} \label{def:f-ef-pool}
An allocation $\mathcal{A}$ is {\em $\rho$-approximate feasible envy-free towards unallocated items ($\rho$-FEFu)} if, for every agent $i \in [n]$ we have $f_i(A_i) \geq \rho\, f_i(S)$ for any $S \subseteq M \setminus \bigcup_{j=1}^{n} A_j$ that is feasible for agent $i$.
\end{definition}

In the next sections, we show that very simple greedy policies in Protocol \ref{alg:MRR} can guarantee value of at least $\optm_i / \Theta(n +p_i)$ to any agent $i\in[n]$ who follows them. This fact implies novel \efo-type results in constrained fair division but also allows for a randomized protocol where the corresponding ex-ante worst-case guarantees are in terms of $\opt_i$ instead. 
Furthermore, it allows us to go beyond worst-case analysis and obtain much stronger guarantees for $(\Omega(n), O(1))$-robust instances that are almost best-possible in polynomial time. 
The approach of these latter results can lead to solid \textit{maximin share} guarantees with small modifications to Protocol \ref{alg:MRR} and to the greedy policies.

\section{The Effectiveness of Greedy Policies for Monotone Objectives}
\label{sec:positive_monotone} 
 We first turn to monotone submodular objective functions; we deal with non-monotonicity in Section \ref{sec:non-monotone}.
 Simple greedy algorithms, where in every step a feasible item of maximum marginal utility is added to the current solution, have found extreme success in a wide range of submodular maximization problems. Hence, it is only natural to consider the following questions:
\begin{quote}
\emph{What value can an agent guarantee for themselves if they always choose greedily? \\ How do agents compare each other's bundles of items in this case? }
\end{quote}
Next we show that an agent $i$ can achieve strong bounds with respect to $\optm_i$. ``Strong'' here refers to the fact that both Theorems \ref{thm:greedy_monotone_general} and \ref{thm:greedy_monotone_cardinality} achieve constant factor approximations to $\optm_i$ for constant $n$,  but also that, for $n=1$, Theorem \ref{thm:greedy_monotone_general} recovers the best-known guarantee of the greedy algorithm for the standard algorithmic problem \citep{FNemhauserW78}, which is  almost best-possible for polynomially many queries \citep{BadanidiyuruV14}. 
Further, for $(\alpha, \beta)$-robust instances these guarantees improve the larger $\alpha$ becomes. In particular,
if an instance is $(\Omega(n), O(1))$-robust with respect to agent $i$, then $i$ achieves a $1/\Theta(p_i)$ fraction of their \textit{optimal} value $\opt_i$ instead, which is asymptotically best-possible.

Formally, when we say that agent $i$ \textit{chooses greedily} we mean that they choose according to the policy $\mathcal{G}_i$ below.

\makeatletter
\renewcommand{\ALG@name}{Policy}
\makeatother

\setcounter{algorithm}{0}

\begin{algorithm}[ht]
		\caption{Greedy policy $\mathcal{G}_i(S_i\,; Q)$ of agent $i$. \\ {\small {($S_i$\,: current solution of agent $i$ (initially $S_{i} =\emptyset$)\,; $Q$\,: the current set of available items)}}}
		\begin{algorithmic}[1]
			\vspace{2pt}\State $A =\{x\in Q\,:\, S_i \cup \{x\} \in \mathcal{I}_i\}$
            \If{$A \neq \emptyset$}
            \State $S_i = S_{i} \cup \{j\}$, where $j\in \argmax_{z\in A} f_i(z\,|\,S_i)$
            \State \Return{$j$}
			\Else
			\State \Return{ a dummy item} {\small\hfill (i.e., return nothing)}
			\EndIf
		\end{algorithmic}
		\label{alg:greedy_policy}
\end{algorithm}\medskip

\begin{theorem}\label{thm:greedy_monotone_general}
Any agent $i$ with a $p_i$-system constraint, who chooses greedily in the Round-Robin protocol, builds a solution $S_i$ such that $f_i(S_i) \ge {\optm_i}/{(n + p_i)}$.
\end{theorem}

\begin{proof}
Let $O_i^-$ be an optimal solution for agent $i$ on the set $M_i$ of items still available after $i-1$ steps, i.e., right before agent $i$ chooses their very first item; by definition, $\optm_i = f_i(O_i^-)$. From $i$'s perspective, it makes sense to consider the $k$-th round to last from when they get their $k$-th item until right before they choose their $(k+1)$-th item. We rename the items of $M_i$ accordingly as $x_1^i, x_1^{i+1}, \ldots, x_1^{i-1}, x_2^i, \ldots$, i.e., item $x_j^{\ell}$ is the $j$-th item that agent $\ell$ chooses from the moment when agent $i$ is about to start choosing; any items not picked by anyone (due to feasibility constraints) are arbitrarily added to the end of the list.
Also, let $S_i^{(r)}$ denote the solution of agent $i$ right before item $x_r^i$ is added to it. Finally, set $s\vcentcolon= |S_i|$.
We are going to need the existence of a mapping $\delta: O_i^-\setminus S_i \to S_i$ with ``nice'' properties, as described in the following lemma for $a = 1$ and $Q_{i1} = O_i^-$; the general form of Lemma \ref{lem:mapping-monotone-general} is needed for the proof of Theorem \ref{thm:greedy_monotone_general_robust}.

\begin{lemma}\label{lem:mapping-monotone-general}
    Let $Q_{i1}, \ldots, Q_{ia} \in \mathcal{I}_i$ be disjoint feasible sets for agent $i$. There is a mapping $\delta:  \bigcup_{j\in[a]} Q_{ij} \setminus S_i \to S_i$ with the following two properties, for all $x_r^i \in S_i, x\in  \bigcup_{j\in[a]} Q_{ij} \setminus S_i$:
\begin{enumerate}[leftmargin=20pt]
    \item if $\delta(x) = x_r^i$, then $f_i(x \,|\, S_i^{(r)}) \le f_i(x_r^i \,|\, S_i^{(r)})$, i.e., $x$ is not as attractive as $x_r^i$ when the latter is chosen;\vspace{5pt}
    \item $|\delta^{-1}(x_r^i)| \le n + a p_i - 1$, i.e., at most $n + a p_i - 1$ items of\, $ \bigcup_{j\in[a]} Q_{ij} \setminus S_i$ are mapped to each item of $S_i$.
\end{enumerate}
\end{lemma}

For now, we assume the lemma (its proof is given right after the current proof) 
and apply the second part of Theorem \ref{thm:SM-monotone} for $O_i^-$ and $S_i$:
\begin{align*}
f_i(O_i^-)  &\le f_i(S_i) + \sum_{x \in O_i^- \setminus S_i} f_i(x\,|\,S_i) 
            = f_i(S_i) + \sum_{r=1}^{s}\sum_{x \in \delta^{-1}(x_r^i)} f_i(x\,|\,S_i)\\
            &\le f_i(S_i) + \sum_{r=1}^{s}\sum_{x \in \delta^{-1}(x_r^i)} f_i(x\,|\,S_i^{(r)})
            \le f_i(S_i) + \sum_{r=1}^{s} (n + p_i - 1) f_i(x_r^i \,|\, S_i^{(r)})\\
            &= f_i(S_i) + (n + p_i - 1) f_i(S_i) 
            = (n + p_i)  f_i(S_i) \,,
\end{align*}
where the first equality follows from observing that $\bigcup_{r=1}^{s} \delta^{-1}(x_r^i) = O_i^-\setminus S_i$, the second inequality follows from submodularity, and, finally, the third inequality follows from the second property of Lemma \ref{lem:mapping-monotone-general}.
We conclude that $f_i(S_i) \ge \optm_i / (n + p_i)$.
\end{proof}

In the proof of Theorem \ref{thm:greedy_monotone_general} one has to deal with the fact that an agent may miss “good” items not only by their own choices (that make such items infeasible), but also because others take such items in a potentially adversarial way. This complication calls for very careful mapping in Lemma \ref{lem:mapping-monotone-general} which concentrates most of the technical difficulty of the proof of the theorem. Note that the lemma allows for a mapping from the union of multiple disjoint feasible subsets of an agent $i$, not just $O_i^-$. This will come handy for proving Theorem \ref{thm:greedy_monotone_general_robust} but adds an extra layer of complexity in keeping track of different kinds of items in its proof.

Before giving the proof of Lemma \ref{lem:mapping-monotone-general}, it is useful to outline the general idea, at least for the special case of a single set, $O_i^-$, as it is used in the proof of Theorem \ref{thm:greedy_monotone_general}. An agent $i$ may miss items from their optimal solution $O_i^-$ for two reasons: (i) because of greedy choices that make them infeasible, and (ii) because other agents take them. It is easy to see that between any two greedy choices of $i$, there are at most $n-1$ such lost items of type (ii) and all have marginals that are no better than the marginal of the former choice. The next step is to show that after the first $\ell$ greedy choices, for any $\ell$, there are at most $\ell p_i$ lost items of type (i). Finally, we argue that this allows us to map at most $p_i$ of them to each greedy choice of $i$ so that they never have a larger marginal value than that greedy choice.

\begin{proof}[\textbf{Proof of Lemma \ref{lem:mapping-monotone-general}}]
We keep the notation and terminology as in the proof of Theorem \ref{thm:greedy_monotone_general}. We are going to define a partition of  $\bigcup_{j\in[a]} Q_{ij}\setminus S_i$. First, let 
$M_i^{(r)} = \{x_r^i, x_r^{i+1}, \ldots, x_r^{i-1}, x_{r+1}^i, \ldots\}$ be the set of available items right before $x^i_r$ is added to $S_i^{(r)}$. Now, for $r \in [s]$, let
\begin{align*}
X_{(r)} =  \big\{&x\in \big(\textstyle \bigcup_{j\in[a]} Q_{ij}\setminus S_i\big) \cap M_i^{(r)} :\, S_i^{(r)}\cup\{x\} \in \mathcal{I}_i 
          \text{\ \ and\ \ } S_i^{(r)}\cup\{x^i_r, x\} \notin \mathcal{I}_i \big\}\,,
\end{align*}
i.e., $X_{(r)}$ contains all items of $ \bigcup_{j\in[a]} Q_{ij}\setminus S_i$ that were available and could be added to $S_i^{(r)}$ right before $x^i_r$ was added, but become infeasible right after. 

Also, for $r \in [s]$ (using the convention that $M_i^{(s+1)} = \emptyset$), let 
\begin{align*}
Y_{(r)} =  \big\{&x\in \big( \big(\textstyle \bigcup_{j\in[a]} Q_{ij}\setminus S_i\big) \cap M_i^{(r)}\big) \setminus \big(X_{(r)} \cup M_i^{(r+1)} \big):\, 
           S_i^{(r+1)}\cup\{x\} \in \mathcal{I}_i \big\}\,, 
\end{align*}
i.e., $Y_{(r)}$ contains all the items of $ \bigcup_{j\in[a]} Q_{ij}\setminus S_i$ that were available and could be added to $S_i^{(r)}\cup\{x^i_r\}$ right after $x^i_r$ was added but were picked by other agents before agent $i$ got to choose for the $(r+1)$-th time. It is easy to see that all these sets are disjoint and that 
\[
\textstyle \bigcup\limits_{j\in[a]} Q_{ij}\setminus S_i = \textstyle \bigcup\limits_{r=1}^{s}(X_{(r)} \cup Y_{(r)})
\,.\]
Also, for any $r \in [s]$, because any $x\in X_{(r)} \cup Y_{(r)}$ was available and feasible when $x_r^i$ was chosen, we have $f_i(x \,|\, S_i^{(r)}) \le f_i(x_r^i \,|\, S_i^{(r)})\le f_i(x_j^i \,|\, S_i^{(j)})$ for any $j\le r$. Our plan is to define $\delta(x)$ so that it maps all $x\in X_{(r)} \cup Y_{(r)}$ to one of $x_1^i,\ldots,x_r^i$ and directly get the first claimed property of $\delta$. To do so in a way that guarantees the second claimed property, we first need to carefully argue about the cardinalities of these sets.

We claim that $|Y_{(r)}|\le n-1$ for all $r \in [s]$. Of course, this is straightforward for $r< s$. For $r = s$, note that
$|Y_{(s)}|\ge n$ would mean that
in the next round of the protocol there will be at least one item left for agent $i$ to add to $S_i$,
contradicting the fact that $s$ is the cardinality of $i$'s solution.

A bound on the cardinality of $X_{(r)}$ is less straightforward. We claim that $|\bigcup_{j=1}^{r} X_{(j)}| \le r a p_i$ for all $r \in [s]$. Suppose not, and let $k$ be the smallest index for which $|\bigcup_{j=1}^{k} X_{(j)}| > k a p_i$. By the  pigeonhole principle, there exists some $\ell\in[a]$, such that $|\bigcup_{j=1}^{k} (X_{(j)}\cap Q_{i\ell})| > k p_i$. 
Recall that $S_i^{(k+1)} = S_i^{(k)}\cup x^i_k$ and that $|S_i^{(k+1)}|=k$. 
By construction, $S_i^{(k+1)}\in \mathcal{I}_i$. Moreover,  by the definition of the $X_{(r)}$'s, $S_i^{(k+1)}$ must be maximally independent in the set $S_i^{(k+1)} \cup \bigcup_{j=1}^{k} (X_{(j)} \cap Q_{i\ell})$. 
Indeed, each $x \in \bigcup_{j=1}^{k} (X_{(j)} \cap Q_{i\ell})$ belongs to some $X_{(j)}$ for $j\in[k]$ and, thus, 
$S_i^{(j)}\cup\{x^i_j, x\} \notin \mathcal{I}_i$. Since $S_i^{(j)}\cup\{x^i_j\} \subseteq S_i^{(k+1)}$, we also have that $S_i^{(k+1)}\cup\{x\} \notin \mathcal{I}_i$. This implies that 
\[\textrm{lr}\big(S_i^{(k+1)} \cup \textstyle \bigcup\limits_{j=1}^{k} (X_{(j)} \cap Q_{i\ell})\big) \le |S_i^{(k+1)}|=k\,. \]
On
the other hand, $\bigcup_{j=1}^{k} (X_{(j)} \cap Q_{i\ell}) \in \mathcal{I}_i$ because $\bigcup_{j=1}^{k} (X_{(j)} \cap Q_{i\ell}) \subseteq Q_{i\ell} \in \mathcal{I}_i$. As a result 
\[\textrm{ur}\big(S_i^{(k+1)} \cup \textstyle \bigcup\limits_{j=1}^{k} (X_{(j)} \cap Q_{i\ell}) \big) \ge \Big|\bigcup\limits_{j=1}^{k} (X_{(j)} \cap Q_{i\ell})\Big|> kp_i\,,\] 
contradicting the fact that $\mathcal{I}_i$ is a $p_i$-system for $M$. We conclude that $|\bigcup_{j=1}^{r} X_{(j)}| \le r a p_i$ for all $r \in [s]$.

Now, we can define $\delta$. We first rename the elements of $\bigcup_{j=1}^{s} X_{(j)}$ so that all items in $X_{(j)}$ are lexicographically before any item in $X_{(j')}$ for $j<j'$. We map the first $a p_i$ items of $\bigcup_{j=1}^{s} X_{(j)}$ (with respect to the ordering we just mentioned) as well as all the items in $Y_{(1)}$ to $x_1^i$, the next $a p_i$ items of $\bigcup_{j=1}^{s} X_{(j)}$ as well as all the items in $Y_{(2)}$ to $x_2^i$, and so forth, until the items of $ \bigcup_{j\in[a]} Q_{ij}\setminus S_i$ are exhausted. Since at most $a p_i+n-1$ items are mapped to each element of $S_i$, both the desired properties of $\delta$ hold.
\end{proof}

Specifically for cardinality constraints (which are $1$-system constraints and include the unrestricted case when the cardinality constraint is at least $\lceil m/n \rceil$), we can get a slightly stronger version of Theorem \ref{thm:greedy_monotone_general}. Note that Theorem \ref{thm:greedy_monotone_cardinality} assumes that $n\ge 2$, and thus it does not contradict the known inapproximability result of \citet{Feige98} (see {Theorem \ref{thm:Feige-hardness})} for the case of a single objective function.
Here the analysis of the greedy policy for Protocol \ref{alg:MRR} is tight for any $i$, as can be seen from the example from the introduction (e.g., for $m=n$).

\begin{theorem}\label{thm:greedy_monotone_cardinality}
For $n\ge 2$, any agent $i$ with a cardinality constraint, who chooses greedily in the Round-Robin protocol, builds a solution $S_i$ such that $f_i(S_i) \ge \optm_i / n$.   
\end{theorem}

\begin{proof}
The proof closely follows that of Theorem \ref{thm:greedy_monotone_general}, so we are just highlighting the differences here.
Consider an agent $i$ with a cardinality constraint of $k$. The second property of the mapping 
$\delta$
now is stronger. In particular, $|\delta^{-1}(x_r^i)| \le \max(n  - 1,a)$ (rather than having $n+ a -1$ on the right-hand side as expected from the proof of Theorem \ref{thm:greedy_monotone_general}). This means that, when $a = 1$ with $Q_{i1} = O_i^-$, applying Theorem \ref{thm:SM-monotone} for $O_i^-$ and $S_i$, we get $f_i(O_i^-)  \le f_i(S_i) + (n  - 1) f_i(S_i)= n\,  f_i(S_i)$ exactly the same way as we got the approximation guarantee in  the proof of Theorem \ref{thm:greedy_monotone_general}. The only thing remaining is the existence of $\delta$, which we show below. 
Note that again we do this for the case of multiple  sets, i.e., $\delta:  \bigcup_{j\in[a]} Q_{ij} \setminus S_i \to S_i$, as we will  need it for  $a = n$ in Section \ref{sec:mms}. 

Like before, set $s\vcentcolon= |S_i|$. For $r \in [s]$, we define $X_{(r)}$ and $Y_{(r)}$ like before, but now notice that $X_{(r)} = \emptyset$ if $r< k$.
Recall that for all $r \in [s]$, any $x\in X_{(r)} \cup Y_{(r)}$ was available and feasible when $x_r^i$ was chosen, so we have $f_i(x \,|\, S_i^{(r)}) \le f_i(x_r^i \,|\, S_i^{(r)})\le f_i(x_j^i \,|\, S_i^{(j)})$ for any $j\le r$. As long as $\delta$ maps each $x\in X_{(r)} \cup Y_{(r)}$ to one of $x_1^i,\ldots,x_r^i$, we only need to worry about keeping $|\delta^{-1}(x_r^i)|$ bounded.

If $X_{(s)} = \emptyset$, then $\bigcup_{j\in[a]} Q_{ij}\setminus S_i = \bigcup_{r=1}^{s}(X_{(r)} \cup Y_{(r)}) = \bigcup_{r=1}^{s} Y_{(r)}$
and $\delta$ is obvious: just mapping all items of $Y_{(r)}$ to $x_r^i$ for $r \in [s]$, suffices in order to get both desired properties. 
So, suppose that $|X_{(s)}| = k' \in [ak]$. Clearly, this means that $X_{(s)} \neq \emptyset$ and thus, $s=k$. 
We begin with setting $\delta(x) = x_r^i$ if $x\in Y_{(r)}$, for all $r\in[k]$; so far at most $n-1$ items are mapped to each $x_r^i$. For the items in $X_{(k)}$, say $z_1, z_2, \ldots, z_{k'}$, 
we map the first $\max(n  - 1,a) - |Y_{(1)}|$ such items to $x_1^i$, the next $\max(n  - 1,a) - |Y_{(2)}|$ items to $x_2^i$, and so on, until they are exhausted. The fact that they are all mapped by the time we map at most $\max(n  - 1,a) - |Y_{(k)}|$ items of $X_{(k)}$ to $x_k^i$ follows by observing that we can map up to $k\max(n  - 1,a)$ items this way and we have at most $|\bigcup_{j\in[a]} Q_{ij}| \le ak$ items to map.
Clearly, we never map more than $\max(n  - 1,a)$ items to a single element of $S_i$, completing the proof.
\end{proof}

\subsection{Implications to Constrained Fair Division}
\label{sec:fairness-monotone}
A reasonable question at this point is whether this is a strong benchmark in general. The answer is that \textit{it depends} on our objective.
On one hand, as we will show below, we get an EF1-type guarantee for any agent who chooses greedily, no matter what other agents do. 
On the other hand, 
it is possible for an agent to improve by a linear factor
by choosing items more carefully; we postpone this discussion until Example \ref{ex:opt_vs_greedy}.

Despite the fact that the Round-Robin allocation protocol is known to produce EF1 allocations
for agents with additive valuation functions and no constraints, Theorems \ref{thm:greedy_monotone_ef1}
and \ref{thm:greedy_monotone_cardinality_ef1} below are  far from trivial and have immediate implications to constrained
fair division; see Corollaries \ref{cor:fairness_p-system} and \ref{cor:fairness_cardinality} below.  
Note that Theorem \ref{thm:greedy_monotone_cardinality_ef1}
is tight, as it is known that Round-Robin as an algorithm cannot guarantee more than $0.5$-EF1 for submodular
valuation functions even without constraints \citep{AmanatidisBL0R23}.

\begin{theorem}\label{thm:greedy_monotone_ef1}
Let $i$ be an agent with a $p_i$-system constraint $\mathcal{I}_i$ who chooses greedily in the Round-Robin protocol. 
Also, let $j$ be any agent and $g$ be the first item added to $S_j$. Then, $f_i(S_i) \ge \max_{S\in \mathcal{I}_i|S'_j}f_i(S) / (p_i +2)$, 
where $S'_j = S_j$ if $i<j$, and $S'_j = S_j\setminus \{g\}$ otherwise.
Furthermore, if $M' = M\setminus \cup_{j\in [n]}S_j$, then $f_i(S_i) \ge \max_{S\in \mathcal{I}_i|M'}f_i(S) / (p_i + 1)$.
\end{theorem}

\begin{proof}
Both parts of the statement are corollaries of Theorem \ref{thm:greedy_monotone_general}, albeit with slightly different setup. 

First, let $j$ be any agent other than $i$ and $S'_j$ be as given in the statement. We can view $S_i$ and $S_j$ as the result 
of running the Round-Robin protocol on $S_i \cup S_j$ with just the two agents $i$ and $j$ following the same order and the exact same choices as before. 
From the point of view of agent $i$, all the choices are consistent with the greedy policy $\mathcal{G}_i$. As a result, by 
Theorem \ref{thm:greedy_monotone_general}, 
\[f_i(S_i) \ge \max_{S\in \mathcal{I}_i|S_i \cup S_j}f_i(S) / (p_i + 2) \ge \max_{S\in \mathcal{I}_i|S'_j}f_i(S) / (p_i +2) \,,\] 
where the last inequality follows from the fact that $S'_j \subseteq S_i \cup S_j$.

For the second part of the statement, we do something similar but assuming the existence of an auxiliary \textit{dummy} agent $i'$,
who has the trivial independence system constraint $\mathcal{I}_{i'} = \{\emptyset\}$. Then, we can view $S_i$ and $M'$ as the result 
of running the Round-Robin protocol on $S_i \cup M'$ with just  $i$ and $i'$ following (just for $i$) the same choices as before. 
From the point of view of $i$, these choices are consistent with the policy $\mathcal{G}_i$. Now, however, the sets $Y_{(r)}$ in the proof
of Lemma \ref{lem:mapping-monotone-general} are all empty (as $i'$ never picks any items) rather than having cardinality up to $n-1 = 1$. 
Thus, part \textit{2.}~of Lemma \ref{lem:mapping-monotone-general}, for $n = 2$ and $a=1$, becomes $|\delta^{-1}(x_r^i)| \le  p_i $ instead of $|\delta^{-1}(x_r^i)| \le p_i + 1$, improving the overall factor of Theorem \ref{thm:greedy_monotone_general} by $1$ in this case. 
That is, by Theorem \ref{thm:greedy_monotone_general}, 
\[f_i(S_i) \ge \max_{S\in \mathcal{I}_i|S_i \cup M'}f_i(S) / (p_i + 1) \ge \max_{S\in \mathcal{I}_i|M'}f_i(S) / (p_i + 1) \,,\] 
where again the last inequality is due to $M' \subseteq S_i \cup M'$.
\end{proof}

\begin{theorem}\label{thm:greedy_monotone_cardinality_ef1}
Let $i$ be an agent with a cardinality constraint $\mathcal{I}_i$ who chooses greedily in the Round-Robin protocol. Also, let $j$ be any agent and $g$ be the first item added to $S_j$. Then, $f_i(S_i) \ge 0.5 \max_{S\in \mathcal{I}_i|S'_j}f_i(S)$, where $S'_j = S_j$ if $i<j$, and $S'_j = S_j\setminus \{g\}$ otherwise. 
\end{theorem}

\begin{proof}
The proof closely follows the first part of the proof of Theorem \ref{thm:greedy_monotone_ef1}. In particular, the setup and the notation are the same. 
The only difference is that when we compare $S_i$ and the best feasible subset of $S_i \cup S_j$ from $i$'s point of view, instead of
Theorem \ref{thm:greedy_monotone_general}, we invoke Theorem \ref{thm:greedy_monotone_cardinality} (for $n=2$) and get 
\[f_i(S_i) \ge \max_{S\in \mathcal{I}_i|S_i \cup S_j}f_i(S) / 2 \ge 0.5 \max_{S\in \mathcal{I}_i|S'_j}f_i(S) \,,\] 
where again the second inequality is due to $S'_j \subseteq S_i \cup S_j$.
\end{proof}

By simulating a run of Round-Robin$(\bm{\mathcal{G}})$, i.e., our Protocol \ref{alg:MRR} with greedy policies for all agents, Theorems \ref{thm:greedy_monotone_ef1} and \ref{thm:greedy_monotone_cardinality_ef1} give us the following direct corollaries. \textit{Maximality} in the statements refers to the fact that adding any unallocated item to any agent's set is infeasible.

\begin{corollary}\label{cor:fairness_p-system}
    For agents with monotone submodular valuation functions and $p$-system constraints, we can efficiently find a $\frac{1}{p+2}$-FEF1 and $\frac{1}{p+1}$-FEFu maximal allocation.
\end{corollary}

\begin{corollary}\label{cor:fairness_cardinality}
    For agents with monotone submodular valuation functions and cardinality constraints, we can efficiently find a $0.5$-FEF1 and $0.5$-FEFu maximal allocation. 
\end{corollary}

Similarly to Theorem \ref{thm:greedy_monotone_cardinality_ef1}, Corollary \ref{cor:fairness_cardinality} is tight even without constraints \citep{AmanatidisBL0R23}. We suspect that the corresponding tight result for Corollary \ref{cor:fairness_p-system} would be $p + 1$ instead.

\subsection{Improved Guarantees for Robust Instances}
\label{sec:robust_monotone}
A natural next question now is: \textit{is it possible for an agent to choose items optimally---or at least, in a way that significantly improves over choosing greedily?} In Section \ref{sec:hardness_monotone}, we show that doing so is computationally hard. Before doing so, however, we divert the question to \textit{under what circumstances} is it possible for an agent to choose items (approximately) optimally by choosing greedily. We turn our attention to the special case of robust instances, for which the results of Theorems \ref{thm:greedy_monotone_general} and \ref{thm:greedy_monotone_cardinality} can be significantly improved.

We begin with a parameterized result about a specific agent, which will imply, as a corollary, a strong guarantee for instances that are $(\Omega(n), O(1))$-robust with respect to everyone (Corollary \ref{cor:greedy_monotone_general_robust_2}). Moreover, Theorem \ref{thm:greedy_monotone_general_robust} will be pivotal for obtaining guarantees with respect to the \textit{maximin shares} of the agents in Section \ref{sec:mms}.
Note that $\gamma$ (or even $\beta$) in the following statement can be functions of $n$; in particular, $\gamma n$ is assumed to be integer.

\begin{theorem}\label{thm:greedy_monotone_general_robust}
    Assume that right before agent $i$ is allowed to pick their first item, the remaining instance is $(\gamma n, \beta)$-robust with respect to $i$ who has a $p_i$-system (resp.~cardinality) constraint. If $i$ chooses greedily in the Round-Robin protocol, they build a solution $S_i$ such that $f_i(S_i) \ge {\opt_i}/{\beta(p_i  + 1 + 1/\gamma)}$ (resp.~$f_i(S_i) \ge {\opt_i}/{\beta \max(2, 1 + 1/\gamma)}$).
\end{theorem}

\begin{proof}
We will prove the theorem for a $p_i$-system constraint and discuss the case of a cardinality constraint briefly at the end of the proof.
Our notation will be consistent with the one introduced in the proof of Theorem \ref{thm:greedy_monotone_general}, 
i.e., the set $M_i$ contains the items still available right before agent $i$ chooses their very first item, item $x_j^{\ell}$ 
is the $j$-th item that agent $\ell$ chooses from the moment when agent $i$ is about 
to start choosing, $S_i^{(r)}$ is the solution of agent $i$ right before item $x_r^i$ is added to it, and 
$s\vcentcolon= |S_i|$.

Since the instance is $(\gamma n, \beta)$-robust with respect to $i$ right before $i$ gets to pick their first item, 
at that time there are at least $a \vcentcolon= \gamma n$ disjoint subsets of $M_i$, say $Q_{i1}, \ldots, Q_{i a}$, 
such that $Q_{ij}\in \mathcal{I}_i$ and $\beta \sum_{j\in[a]} f_i(Q_{ij}) \ge a\,\opt_i$.

Here we invoke the full power of Lemma \ref{lem:mapping-monotone-general}, as we apply it to all $Q_{i1}, \ldots, Q_{i a}$ and assume the existence of $\delta$ as in the statement of the lemma.  Further, we apply the second part of Theorem \ref{thm:SM-monotone} for $S_i$ and each $Q_{ij}$, and we add up these $a$ inequalities:
\begin{align*}
\sum_{j=1}^{a} f_i(Q_{ij})  &\le \sum_{j=1}^{a} f_i(S_i) + \sum_{j=1}^{a}\sum_{x \in Q_{ij} \setminus S_i} \!\! f_i(x\,|\,S_i)   
            = a  f_i(S_i) + \!\! \sum_{x \in ( \bigcup_{j\in[a]} Q_{ij}) \setminus S_i} \!\!\! f_i(x\,|\,S_i) \\
            &\le a f_i(S_i) + \sum_{r=1}^{s}\sum_{x \in \delta^{-1}(x_r^i)} f_i(x\,|\,S_i)
            \le a f_i(S_i) + \sum_{r=1}^{s}\sum_{x \in \delta^{-1}(x_r^i)} f_i(x\,|\,S_i^{(r)})\\
            &\le a f_i(S_i) + \sum_{r=1}^{s} (n + a p_i - 1) f_i(x_r^i \,|\, S_i^{(r)})
            = \gamma n f_i(S_i) + (n + \gamma n p_i - 1) f_i(S_i) \\[5pt]
            &\le n (\gamma + 1+ \gamma p_i)  f_i(S_i) \,,
\end{align*}
where the second inequality follows from the fact that $\bigcup_{r=1}^{s} \delta^{-1}(x_r^i) = \big( \bigcup_{j\in[a]} Q_{ij} \big) \setminus S_i$, the third follows from submodularity, and the fourth from the second property of Lemma \ref{lem:mapping-monotone-general}. On the other hand, 
\[
\sum_{j=1}^{a} f_i(Q_{ij})  \ge \frac{a\, \opt_i}{\beta}  = \frac{\gamma n}{\beta} \opt_i\,.
\]
We conclude that
$f_i(S_i) \ge  {\opt_i}/{\beta(p_i  + 1 + 1/\gamma)}$, as claimed.

In the case of a cardinality constraint, by using the analog of Lemma \ref{lem:mapping-monotone-general} implied in the proof of Theorem \ref{thm:greedy_monotone_cardinality}, the proof is exactly the same as above except that in the fourth inequality we get $\max(n-1, a)$ instead of $n + a p_i - 1$. Then the factor is a matter of simple calculations.
\end{proof}

\begin{corollary}\label{cor:greedy_monotone_general_robust_1}
    Assume that an instance is $(\gamma n +i -1, \beta)$-robust with respect to agent $i$ 
    who has a $p_i$-system (resp.~cardinality) constraint. If $i$ chooses greedily in the Round-Robin protocol, they build a solution $S_i$ such that $f_i(S_i) \ge {\opt_i}/{\beta(p_i  + 1 + 1/\gamma)}$ (resp.~$f_i(S_i) \ge {\opt_i}/{\beta \max(2, 1 + 1/\gamma)}$).
\end{corollary}

\begin{proof}
Since the instance is $(\gamma n +i -1, \beta)$-robust with respect to $i$, 
by the time $i$ gets to pick their first item, there are at least $a \vcentcolon= \gamma n$ 
disjoint subsets of $M$, say $Q_{i1}, \ldots, Q_{i a}$, such that $Q_{ij}\in \mathcal{I}_i$ 
and $\beta \sum_{j\in[a]} f_i(Q_{ij}) \ge a\, \opt_i$, \textit{and} no item in 
$\cup_{j=1}^{a} Q_{ij}$ has been allocated to any of the agents $1$ through $i-1$ yet.
That is, the remaining instance is $(\gamma n, \beta)$-robust with respect to $i$ 
and the corollary follows by applying Theorem \ref{thm:greedy_monotone_general_robust}.
\end{proof}

Although Theorem \ref{thm:greedy_monotone_general_robust} and Corollary \ref{cor:greedy_monotone_general_robust_1} are 
agent specific, in applications like our influence maximization running example one would expect that typical real-world 
instances are extremely robust with respect to everyone. The rationale here is that the set $M$ contains much more value 
than what an agent can extract given their constraint. The following corollary gives strong guarantees with respect to all 
agents for such scenarios. Recall that these guarantees are almost the best one could hope for in polynomial 
time \citep{BadanidiyuruV14}.

\begin{corollary}\label{cor:greedy_monotone_general_robust_2}
    Assume that an instance is $(\lceil(1+\gamma) n\rceil , \beta)$-robust with respect to every agent for constant $\beta, \gamma \in \mathbb{R}_+$. Any agent $i$ with a $p_i$-system constraint who chooses greedily in the Round-Robin protocol, achieves a $1/\Theta(p_i)$ fraction of\, $\opt_i$. In particular, for $\beta = \gamma  = 1$, any agent $i$ achieves $\opt_i/(p_i + 2)$ for a $p_i$-system constraint (hence, $\opt_i/3$ for a cardinality constraint).
\end{corollary}

\begin{proof}
We have that $\lceil(1+\gamma) n\rceil = n + \lceil \gamma n\rceil = n + n\zeta(n) $, where $\zeta(n) \vcentcolon= \lceil \gamma n\rceil / n$, for any $n\in \mathbb{N}$, and thus  $\gamma - 1/n < \zeta(n) \le \gamma$.
By Definition \ref{def:robust}, if an instance is $((1+\zeta(n)) n , \beta)$-robust with respect to agent $i$, then it is also $(\zeta(n) n +i -1, \beta)$-robust with respect to $i$. By applying Corollary \ref{cor:greedy_monotone_general_robust_1}, we get $\opt_i \le  \beta(p_i  + 1 + 1/\zeta(n)) f_i(S_i)$, and since $\beta$ and $\gamma$ are constant, $\beta(p_i  + 1 + 1/\zeta(n)) = \Theta(p_i)$. In particular, $\beta(p_i  + 1 + 1/\zeta(n)) = p_i + 2$, for $\beta = \gamma  = 1$. 
\end{proof}

As we mentioned above, our definition of $(\alpha, \beta)$-robustness and Theorem \ref{thm:greedy_monotone_general_robust} lead to simple protocols with approximate \textit{maximin share} guarantees. In particular, in Section \ref{sec:mms}, we show that variants of Algorithm 2 of \citet{AMNS17} and Algorithm 3 of \citet{BarmanK20} (which essentially fall under our notion of a protocol) imply $1/\Theta(p_i)$-approximate guarantees with respect to each agent's feasible maximin share.

\section{It is Hard to Optimize Over Others}
\label{sec:hardness_monotone}
So far, there was no particular need to talk about computational efficiency. The greedy policy clearly runs in polynomial time, assuming polynomial-time value and independence oracles, and the Round-Robin protocol itself delegates any non-trivial computational task to the agents themselves. Here, we need to clarify our terminology a bit. 
When we refer to polynomial-time algorithms in the statements of the next two theorems we mean  algorithms that determine the choices of a single agent, who has full information about all the objective functions, and whose number of steps and number of queries to submodular function value oracles are bounded by a polynomial in the number of agents. Within their proofs, however, we reduce the existence of such algorithms to polynomial-time approximation algorithms for constrained submodular maximization, whose  number of steps and number of queries to a submodular function value oracle are bounded by a polynomial in the size of the ground set over which the submodular function is defined.

When we say that an algorithm $\mathcal{A}_i$ is a \emph{$\rho$-improvement of agent $i$ over the greedy policy $\mathcal{G}_i$ (given the objective functions and policies of all other agents)}, we mean that running $\mathcal{A}_i$, while the other agents stick to their corresponding policies, results in agent $i$ receiving a set $\rho$ times more valuable than the one they receive by choosing greedily, whenever this is possible.

The main result here is Theorem \ref{thm:second_hardness}. On a high level,  we show that in this setting the inapproximability result of \citet{Feige98} is amplified to the point where consistently doing \textit{slightly} better than the greedy policy is NP-hard, even when other agents pick elements in the most predictable way. 
At a first glance, it might not be obvious that the theorem says something nontrivial; our setting generalizes the problem of maximizing a submodular function subject to a cardinality constraint, so it is expected that it is impossible to beat the performance of a greedy solution in polynomial time.

However, this hardness result holds 
even when there are \textit{no individual feasibility constraints} and when all other agents (i.e., everyone except the agent to whom the hardness applies) use strategies as simple as choosing greedily.  That is, the hardness does not stem from having strong constraints or hard to analyze policies for everyone else, but rather indicates the inherent computational challenges of the problem. 
We find this to be rather counter-intuitive.
The proof of Theorem \ref{thm:second_hardness} is not straightforward either. To prove the theorem we do not construct a general reduction as usual, but rather we deal with a number of cases algorithmically and we only use a reduction for instances with very special structure.

\begin{theorem}\label{thm:second_hardness}
Assume $n\ge 2$. Let $\varepsilon \in (0, 0.2)$ be any small constant and $j\in[n]$ be any agent.
Even in instances where all agents in $[n]\setminus\{j\}$ have additive objective functions and greedy policies, there is no polynomial-time algorithm that is a $(1+\varepsilon)$-improvement of agent $j$ over the greedy policy $\mathcal{G}_j$,
unless $\mathrm{P} = \mathrm{NP}$.
\end{theorem}

\begin{proof}
Fix some $\varepsilon\in (0, 0.2)$.
For simplicity, we are going to show the result for $j = 1$, i.e., for agent $1$; the proof is analogous for general $j$. 
We show that the existence of such an algorithm, call it $\mathcal{A}_1$ 
(which is also a function of $\mathcal{A}_2, \ldots, \mathcal{A}_n$), 
implies the existence of a polynomial-time $(1- \frac{1}{e} + \frac{\varepsilon}{10})$-approximation algorithm for maximizing a submodular function subject to a cardinality constraint, directly contradicting the classic hardness result of \citet{Feige98}:

\begin{theorem}[\citet{Feige98}]\label{thm:Feige-hardness}
    Unless\hspace{.1em} $\mathrm{P} = \mathrm{NP}$, for any constant $\delta > 0$, there is no polynomial-time $(1- {1}/{e} + \delta)$-approximation algorithm for maximizing a submodular function subject to a cardinality constraint.
\end{theorem}

Given a normalized, monotone submodular function $f:2^A\to \mathbb{R}_{\ge0}$, where $|A| = n\ge 2$, and a cardinality constraint $k\in [n]$, let $\textsc{Opt} \in \argmax_{S\subseteq A, |S|= k} f(S)$ and $\textsc{Gre}$ be the solution constructed by the greedy algorithm that starts with $\textsc{Gre} = \emptyset$ and for $k$ steps selects an item with maximum marginal value with respect to $\textsc{Gre}$ and adds it therein (ties broken lexicographically). It is known that $f(\textsc{Gre})\ge (1- \frac{1}{e}) f(\textsc{Opt})$ \citep{NemhauserWF78}. 
We are going to exploit a simple property of the greedy algorithm on monotone submodular functions, namely, the fact that either the greedy solution is optimal, or the $k$-th element added to it had marginal value at least ${(f(\textsc{Opt})-f(\textsc{Gre}))}/{k}$. This is easy to see: right before the greedy algorithm picked the $k$-th element, $j_k$, there was at least as much available value as $f(\textsc{Opt}\cup\textsc{Gre})-f(\textsc{Gre}\setminus\{j_k\})\ge f(\textsc{Opt})-f(\textsc{Gre})$ split among the items of $\textsc{Opt}\setminus(\textsc{Gre}\setminus\{j_k\})$ (the number of which is at most $k$).

 If $k\le6$ then the problem can be solved optimally in polynomial time with polynomially many value queries by checking all possible subsets of cardinality $k$, so assume that $k\ge 7$.
Consider the  marginal value $\mu$ of the $k$-th element added to $\textsc{Gre}$; as discussed above, this marginal value should be at least $(f(\textsc{Opt})-f(\textsc{Gre}))/{k}$. Of course, if $\mu = 0$, then $\textsc{Gre}$ is an optimal solution. Otherwise (i.e., when $\mu \neq 0$), we consider two cases, depending on how large the ratio $f(\textsc{Gre})/\mu$ is. This distinction allows us later to construct a $\multisub$ instance that is polynomially related to the original instance.

If $\left\lfloor f(\textsc{Gre})/\mu \right \rfloor > n^2$, then
\[f(\textsc{Gre})>   n^2 \mu \ge   n^2 \frac{f(\textsc{Opt})-f(\textsc{Gre})}{k} \ge n (f(\textsc{Opt})-f(\textsc{Gre}))\,,\]
and, thus,
\[f(\textsc{Gre})> \frac{n}{n+1} f(\textsc{Opt}) \ge  \frac{2}{3} f(\textsc{Opt}) > \big(1- \frac{1}{e} + \frac{\varepsilon}{10}\big ) f(\textsc{Opt})\,,\]
where the second inequality follows from the monotonicity of $\frac{x}{x+1}$ and the third inequality reflects the simple fact that $2/3 > 1- e^{-1} + 0.025$.

If $\left\lfloor f(\textsc{Gre})/\mu \right \rfloor \le n^2$, we are going to construct a $\multisub$ instance. In what follows, we assume that $f(\textsc{Gre})/\mu$ is integer. While this simplifies things, it is without loss of generality: since, by definition, $f(\textsc{Gre})\ge \mu$, we have that $f(\textsc{Gre})/\mu >1$ and we could use $\hat{\mu} = f(\textsc{Gre}) / \left\lfloor f(\textsc{Gre})/\mu \right \rfloor$ instead of $\mu$ in the definition of $f_1$ below without changing anything else, besides replacing $f(\textsc{Gre})/\mu $ with $\left\lfloor f(\textsc{Gre})/\mu \right \rfloor$ and minimally modifying the lower bounding of $f_1(S_1^+[k]\cap A)$.

So, consider a $\multisub$ instance with $n$ agents and $m=kn+2kn\lambda  f(\textsc{Gre})/\mu$ items, where $\lambda = \big\lfloor (\varepsilon(e-1))^{-1}-2\big\rfloor$.
Let the set of items be $A\cup D\cup L$, where $D$ is going to be a set of $(k-1)n$ auxiliary items of low value for agent $1$ and $L$ is a set of $2kn\lambda  f(\textsc{Gre})/\mu $ auxiliary items of low value for everyone. For the objective functions, we have
\[f_1(S)  =  f(S\cap A) +\frac{\mu}{2k}|S\cap L|\,
\text{\ \ \ and\ \ \ \,}
f_i(S)  =  2|S\cap D| + |S\cap A| + \frac{1}{2} |S\cap L| \,, \text{ for } 2\le i\le n \,.\]
Note that $f_1$ is normalized, monotone submodular and $f_2, \ldots, f_n$ are additive.
We further assume that items in $A$ lexicographically precede items in $D$ (this is only relevant when $\textsc{Gre}$ contains items of zero marginal value).

Now, if every agent $i\in[n]$ uses their greedy policy $\mathcal{G}_i$ here, then for the allocation $(S_1, S_2, \ldots, \allowbreak S_n)$ returned by the protocol Round-Robin$(\bm{\mathcal{G}})$, where $\bm{\mathcal{G}} = (\mathcal{G}_1,  \ldots, \mathcal{G}_n)$, we have:
\begin{enumerate}[leftmargin=23pt]
    \item All items in $A\cup D$ are allocated during the first $k$ rounds, whereas all items in $L$ are allocated during the subsequent $2k \lambda f(\textsc{Gre})/\mu$ rounds.
    \item $S_1 = \textsc{Gre}\cup S_L$, where $S_L\subseteq L$ and $|S_L| = 2k \lambda f(\textsc{Gre})/\mu$, and therefore, $f_1(S_1) = (\lambda +1) f(\textsc{Gre})$;
    \item All items in $A\setminus \textsc{Gre}$ are allocated to agents $2, \ldots, n$ in round $k$, after the whole $\textsc{Gre}$ has already been allocated to agent $1$.
\end{enumerate}
The idea here is to allow agent $1$ to get any $k$ items from $A$ before these get allocated to the other agents during round $k$; then everybody builds more value via the low-valued items of $L$ for the rest of the execution of Round-Robin$(\bm{\mathcal{G}})$.
If $f_1(\textsc{Gre}) = f(\textsc{Gre}) < f(\textsc{Opt})$, then any algorithm $\mathcal{A}'_1$ that chooses the items in $\textsc{Opt}$ before any other items is an improvement for agent $1$ over $\mathcal{G}_1$; Running Round-Robin$(\mathcal{A}'_1, \bm{\mathcal{G}}_{-1})$---where $\bm{\mathcal{G}}_{-1}$ denotes $(\mathcal{G}_2,  \ldots, \mathcal{G}_n)$ as usual---results in agent $1$ receiving $\textsc{Opt}$ along with $2k \lambda f(\textsc{Gre})/\mu$ items from $L$ and, thus, obtaining value $f(\textsc{Opt}) + \lambda  f(\textsc{Gre})$ instead.

If $f(\textsc{Opt}) + \lambda  f(\textsc{Gre}) \le (1 + \varepsilon) (\lambda +1) f(\textsc{Gre})$, then
\begin{align*}
f(\textsc{Opt}) \le (1+ \lambda \varepsilon +\varepsilon) f(\textsc{Gre}) 
		\le \big(1+ \frac{1}{e-1}-2\varepsilon +\varepsilon\big) f(\textsc{Gre}) 
		\le \big(\frac{e}{e-1} - \varepsilon\big) f(\textsc{Gre}) \,,
\end{align*}
where the second inequality directly follows by the definition of $\lambda$. Therefore,
\[f(\textsc{Gre}) \ge \big(\frac{e}{e-1} - \varepsilon\big)^{-1} f(\textsc{Opt}) > \big(1- \frac{1}{e} + \frac{\varepsilon}{10}\big ) f(\textsc{Opt})\,,\]
where it is easy to check that the second inequality holds for all $\varepsilon \in (0, 0.2)$.

Finally, if $f(\textsc{Opt}) + \lambda  f(\textsc{Gre}) > (1 + \varepsilon) (\lambda +1) f(\textsc{Gre})$, then  $\mathcal{A}'_1$ is a $(1+\varepsilon)$-improvement of agent $1$ over $\mathcal{G}_1$.
In such a case, by its definition, $\mathcal{A}_1$ is also a (possibly different) $(1+\varepsilon)$-improvement of agent $1$ over $\mathcal{G}_1$. Running Round-Robin$(\mathcal{A}_1, \bm{\mathcal{G}}_{-1})$ results in agent $1$ receiving a solution $S_1^+$, such that
\begin{align*}
f_1(S_1^+) 	&\ge (1+\varepsilon) (\lambda +1) f(\textsc{Gre}) \\
		&\ge \lambda  f(\textsc{Gre}) + (1+\lambda\varepsilon+ \varepsilon)\Big(1- \frac{1}{e}\Big) f(\textsc{Opt}) \,,
\end{align*}
where the first inequality follows from the choice of $\mathcal{A}_1$ and $S_1^+$, and the second from simple algebra and the guarantee of the greedy algorithm. Note that the second term of the right-hand side must correspond to the items added to $S_1^+$ in the first $k$ rounds of Round-Robin$(\mathcal{A}_1, \bm{\mathcal{G}}_{-1})$, as the greedy choices of agents $2, \ldots, n$ guarantee that from round $k+1$ onward only items from $L$ will be allocated (resulting in the $\lambda  f(\textsc{Gre})$ term of the right-hand side).  Let us denote  by $S_1^+[k]$ these items added to $S_1^+$ in the first $k$ rounds. We are interested in the value of $S_1^+[k]\cap A$. Note that, trivially, $|S_1^+[k]\cap L|\le k$, and thus
\begin{align*}
f_1(S_1^+[k]\cap A)   &\ge (1+\lambda\varepsilon+ \varepsilon)\,\frac{e-1}{e}\, f(\textsc{Opt}) - f_1(S_1^+[k]\cap L) \\
                        &\geq \Big(1+ \Big( \frac{1}{\varepsilon(e-1)}-3 \Big) \varepsilon + \varepsilon\Big) \,\frac{e-1}{e} \, f(\textsc{Opt}) - k\, \frac{\mu}{2k} \\
                        &\geq \Big(\frac{e}{(e-1)}-2 \varepsilon \Big) \,\frac{e-1}{e} \, f(\textsc{Opt}) - \frac{f(\textsc{Opt})}{2k} \\
                        &= \Big(1- \frac{2 \varepsilon(e-1)}{e} - \frac{1}{2k}\Big) \, f(\textsc{Opt}) \\
                        &\geq \Big(1- \frac{2 (e-1)}{5e} - \frac{1}{14}\Big) \, f(\textsc{Opt}) \\
                        &\geq \big(1- \frac{1}{e} + \frac{\varepsilon}{10}\big ) f(\textsc{Opt}) \,,
\end{align*}
where the second inequality follows by using $(\varepsilon(e-1))^{-1}-3$ as a lower bound for $\lambda$, the third inequality uses a trivial upper bound for $\mu$, namely $\mu\le f(\textsc{Opt})/k$, the fourth inequality uses the fact that $\varepsilon<1/5$ and $k\ge 7$, and finally it is a matter of simple calculations to check that the last inequality holds for all $\varepsilon\in (0, 0.2)$.

Thus, besides the case where $k$ is small and we optimally solve the problem using brute force, either $\textsc{Gre}$ or $S_1^+[k]$ guarantee a value of at least  $(1- \frac{1}{e} + \frac{\varepsilon}{10}) f(\textsc{Opt})$, leading to the following polynomial-time $(1- \frac{1}{e} + \frac{\varepsilon}{10})$-approximation algorithm for maximizing a submodular function $f$, subject to a cardinality constraint $k$:
\begin{itemize}
    \item If $k\le 6$ use brute force to solve the problem optimally by considering all $k$-subsets of $A$;
    \item Otherwise, run the greedy algorithm to find a solution $\textsc{Gre}$ and use it to calculate $\mu$. If $\mu = 0$ or $f(\textsc{Gre})/\mu > n^2$, then return $\textsc{Gre}$;
    \item Otherwise, construct the $\multisub$ instance described above;
    \item Run Round-Robin$(\mathcal{A}_1, \bm{\mathcal{G}}_{-1})$ to obtain the solution $S_1^+$, and thus, $S_1^+[k]$;
    \item If $f(S_1^+[k]\cap A) >  f(\textsc{Gre})$, then return $S_1^+[k]$, else return $\textsc{Gre}$.
\end{itemize}
As there is no such algorithm \citep{Feige98}, there is no polynomial-time algorithm that returns a $(1+\varepsilon)$--improvement for agent $1$ over $\mathcal{G}_1$ whenever such an improvement exists.
\end{proof}

Despite Theorem \ref{thm:second_hardness}, as we mentioned earlier and as it is illustrated by the following example, 
there are instances where an agent can improve their value by a factor of order $\Omega(n)$ relative to the greedy solution
by choosing items more carefully and taking other agents' objective functions into consideration. 
Another way to interpret Example \ref{ex:opt_vs_greedy} is that Protocol \ref{alg:MRR}, viewed as a mechanism, is not truthful and that its \textit{incentive ratio} (i.e., the ratio of the optimal value an agent can achieve by deviating to the value they get by reporting truthfully) is $\Omega(n)$.

\begin{example}\label{ex:opt_vs_greedy}
    Consider an instance with $n\ge 3$ agents and a set $M=\{g_1, g_2, \ldots, g_m\}$ of $m=n^2+1$ items. The agents here have no combinatorial constraints. We begin by describing the objective function $f_1(\cdot)$ of agent 1. We define $f_1(\cdot)$ to be additive, and set $f_1(g_j)=1$, if $g_j \in \{g_1,\ldots, g_{2n}\}$, and $0$ otherwise. Let $1 \gg \epsilon_1>\epsilon_2>\epsilon_3>\epsilon_4>0$.  For the rest of the agents, i.e., for every agent $i \neq 1$,  we define their objective functions as follows:
\[
	f_i(g_j)=  
	\begin{cases}
	1+\epsilon_1, & \text{if $j=i-1$} \\
	1 + \epsilon_2, & \text{if $j=i$}  \\
	1+ \epsilon_3, & \text{if $j=2 i$} \\
        1 + \epsilon_4, & \text{if $j \geq 3 n+1$} \\
        0, & \text{otherwise}
        \end{cases} 
\]
Moreover, for any set $Z\subseteq  M\setminus \{g_{i-1}\}$ or $Z\subseteq M\setminus \{g_{2  i}\}$, we have that $f_i(Z)=\sum_{j \in Z}f_i(j)$. Otherwise, i.e., when $\{g_{i-1}, g_{2  i}\} \subseteq Z$, we have $f_i(Z)=\sum_{j \in Z} f_i(j)- \epsilon_3$. Functions $f_i(\cdot)$ are now well-defined. 

We continue by proving that for every $i \neq 1$, the objective function $f_i(\cdot)$ is submodular. For this, consider arbitrary sets $S \subseteq T \subseteq M$, and distinguish the following cases:
\begin{enumerate}[leftmargin=23pt]
    \item $\{g_{i-1}, g_{2  i}\} \cap S=\emptyset$ and $|\{g_{i-1}, g_{2  i}\} \cap T|=1$: If $g_{i-1}\in T$ (resp.~$g_{2i}\in T$), then $f_i(g_{2i}\,|\,S)> f_i(g_{2i}\,|\,T)$ (resp.~$f_i(g_{i-1}\,|\,S)> f_i(g_{i-1}\,|\,T)$), and $f_i(x\,|\,S)= f_i(x\,|\,T)$, for any \emph{other} $x \notin T$.
    \item Any other configuration of $S$, $T$, $g_{i-1}$, and $g_{2  i}$: it is not hard to see that $f_i(x\,|\,S)= f_i(x\,|\,T)$, for any  $x \notin T$.
\end{enumerate}

Now our focus will be on agent $1$. Consider a run of the Round-Robin protocol where every agent plays greedily, and let every agent break ties lexicographically.\footnote{\ This assumption is without loss of generality. The given instance can be easily modified (at the expense of simplicity) to one with no ties, and thus no specific tie breaking rule is needed.} It is easy to verify that, in this case, agent 1 will get only items $g_1$, and $g_{n+1}$ from $\{g_1,\ldots, g_{2  n}\}$, something that implies that their total attained value in the end will be $f_1(S_1) = 2$.

On the other hand, a seemingly naive upper bound on the value that agent $1$ could secure through any policy is $n+1$ (and, clearly, $\opt_1=2n$). 
To see this, simply notice that agent $1$ can get $n+1$ items and the most optimistic scenario is that all of them have a marginal value of $1$.
Actually, and somewhat surprisingly, agent $1$ can attain this value, given that the other agents choose their items greedily, by implementing the following policy: In the first round agent $1$ picks item $g_n$ instead of $g_1$, thus leading every agent $i \neq 1$ to get as their first choice item $g_{i-1}$ instead of $g_i$. By the way function $f_i(\cdot)$ is defined, 
for $S \supseteq \{g_{i-1}\}$, if $x \in \{g_{2  n},\ldots, g_{3   n-1}\}$ we have that $f_i(x\,|\,S)$ is either $1$ or $0$, while if $x \in \{g_{3  n},\ldots,g_m\}\setminus S$ we have that $f_i(x\,|\,S)=1+\epsilon_4$. The latter implies that no agent $i\neq 1$ will ever choose something from $\{g_{2  n},\ldots, g_{3   n-1}\}$, leaving all these items available to agent $1$.
This allows agent $1$ to obtain a solution $S$ of maximum possible value, namely $f_1(S) = n+1$. \qed
\end{example}

\section{Dealing With Non-Monotone Objectives}
\label{sec:non-monotone}

When trying to maximize a non-monotone submodular function---with or without constraints---a naively greedy solution may be arbitrarily far from a good solution even in the standard setting with a single agent  (for such an example, see, e.g., \citet{AFLLR22}). 
Several approaches have been developed to deal with non-monotonicity, but 
the recently introduced \emph{simultaneous greedy} approach \citep{AmanatidisKS22,FeldmanHK23} 
seems like a natural choice here. 
The high-level idea is to bypass the complications of non-monotonicity by simultaneously constructing multiple greedy solutions. 
Here we consider building only \emph{two} greedy solutions, not just for the sake of simplicity but also because the technical analysis for multiple solutions does not translate well into our setting.
Although we suspect that a more elaborate approach with more solutions could shave off a constant factor asymptotically, we see the results of this section as a proof of concept that a simple greedy policy can work well in the non-monotone version of our problem as well. 

Coming back to building two greedy solutions, an agent $i$
maintains two sets $S_{i1}, S_{i2}$ and every time that it is their turn to pick an item, they pick a single item that maximizes the marginal value with respect to either $S_{i1}$ or $S_{i2}$ among the items that are still available and for which adding them to the respective solution is feasible. 
Formally, saying that agent $i$ chooses greedily, now means that they choose according to the following policy $\mathcal{G}^+_i$:

\begin{algorithm}[ht]
		\caption{Simultaneous Greedy policy $\mathcal{G}^+_i(S_{i1}, S_{i2}\,; Q)$ of agent $i$. \\ {\small {($S_{i1}, S_{i2}$\,: current solutions of $i$ (initially $S_{i1} = S_{i2} =\emptyset$)\,; $Q$\,: the current set of available items)}}}
		\begin{algorithmic}[1]
            \vspace{2pt}\State $A=\{(x,y)\in Q\times \{1,2\}\,:\, S_{iy} \cup \{x\} \in \mathcal{I}_i\}$
			\If{$A\neq \emptyset$}
            \State Let $(j, \ell)\in \argmax_{(z,w)\in A} f(z\,|\,S_{iw})$
            \State $S_{i\ell} = S_{i\ell} \cup \{j\}$ 
            \State \Return{$j$}
			\Else
			\State \Return{ a dummy item} {\small\hfill (i.e., return nothing)}
			\EndIf
		\end{algorithmic}
		\label{alg:simultaneous_greedy_policy}
\end{algorithm}

\begin{theorem}\label{thm:greedy_non-monotone_general}
Any agent $i$ with a non-monotone objective and a $p_i$-system constraint, who chooses greedily in
the Round-Robin protocol, builds two solutions $S_{i1}, S_{i2}$ so that 
$\max_{t\in\{1, 2\}} f_i(S_{it}) \ge {\optm_i}/{(4n + 4p_i +2)}$.
\end{theorem}

\begin{proof}
The initial setup of the proof is similar to the proof of Theorem \ref{thm:greedy_monotone_general}, but the mapping $\delta$ and the analysis are more involved.
So, let $O_i^-$ be an optimal solution for agent $i$ on the set $M_i$ of items still available after $i-1$ steps; clearly $\optm_i = f_i(O_i^-)$. We rename the items of $M_i$ as $x_1^i, x_1^{i+1}, \ldots, x_1^{i-1}, x_2^i, \ldots$, so that item $x_j^{\ell}$ is the $j$-th item that agent $\ell$ chooses from the moment when agent $i$ is about to start choosing; any items not picked by anyone (due to feasibility constraints) are arbitrarily added to the end of the list. 
Also, let $S_{it}^{(r)}, t\in\{1,2\}$, denote the solutions of agent $i$ right before item $x_r^i$ is added to one of them. 
In particular, $x^i_r$ is added to $S_{i,\tau(r)}^{(r)}$, i.e.,  $\tau(r)$ denotes the index of the solution augmented with $x^i_r$ for any $r\le s$, where $s \vcentcolon= |S_{i1}| +|S_{i2}|$.
Finally, let $Z = \{x\in O_i^- :\, f_i(x\,|\, S_{it}) < 0, \text{\ for\ } t \in \{1, 2\}\}$ contain all the items of $O_i^-$ that have negative marginal values with respect to both final solutions of $i$.

Like in Theorems \ref{thm:greedy_monotone_general} and \ref{thm:greedy_monotone_cardinality}, we need a mapping $\delta$ with similar properties. This is given in the following lemma, which is proved right after the current proof.

\begin{lemma}\label{lem:mapping-non-monotone-general}
    Let $Q_{i1}, \ldots, Q_{ia} \in \mathcal{I}_i$ be disjoint feasible sets for agent $i$, and $P = \{x\in \bigcup_{j\in[a]} Q_{ij} :\, \text{\ there is\ } t \in \{1, 2\} \text{\ such that\ } f_i(x\,|\, S_{it}) \ge 0\}$ be the set of items in $\bigcup_{j\in[a]} Q_{ij}$ that do not have negative marginal values with respect to both final solutions of $i$. There is a mapping $\delta:  P \to S_{i1} \cup S_{i2}$ with the following two properties,  for all $x_r^i \in S_{i1} \cup S_{i2}$ and $x\in P$:
\begin{enumerate}[leftmargin=23pt,itemsep=-1pt]
    \item if $\delta(x) = x_r^i$, then 
    $f_i(x \,|\, S_{it}^{(r)}) \le f_i\big(x_r^i \,|\, S_{i,\tau(r)}^{(r)}\big)$, 
    i.e., $x$ is not as attractive as $x_r^i$ when the latter is chosen;\vspace{5pt}
    \item $|\delta^{-1}(x_r^i)| \le n + a p_i$, i.e., at most $n + a p_i$ items of $P$ are mapped to each item of $S_{i1} \cup S_{i2}$.
\end{enumerate}
\end{lemma}

For now, we need Lemma \ref{lem:mapping-non-monotone-general} for $a = 1$, $Q_{i1} = O_i^-$, and $P = O_i^- \setminus Z$; the general form of the lemma is needed for the proof of Theorem \ref{thm:greedy_non-monotone_general_robust}.
Let $\tau \in \argmax_{t\in[2]} f_i(S_{it})$, $S_i\vcentcolon= S_{i\tau}$, and apply the first part of Theorem \ref{thm:SM-general} for $O_i^-\cup S_{it}$ and $S_{it}$, $t\in\{1,2\}$:
\begin{align*}\label{eq:S_i1}
f_i(O_i^-\cup S_{it})   \le f_i(S_{it}) + \sum_{x \in O_i^- \setminus S_{it}} f_i(x\,|\,S_{it})
                         \le f_i(S_{i}) + \sum_{x \in O_i^-} f_i(x\,|\,S_{it}) \,,
\end{align*}
where we padded the sum with extra terms of zero value. 
We combine those two inequalities to get
\begin{align*}
\sum_{t\in[2]} f_i(O_i^-\cup S_{it}) &\le 2f_i(S_{i}) + \sum_{x \in O_i^-} \sum_{t\in[2]} f_i(x\,|\,S_{it}) 
            \le 2f_i(S_{i}) + \sum_{x \in O_i^-\setminus Z} \sum_{t\in[2]} f_i(x\,|\,S_{it}) \\
            &\le 2f_i(S_{i}) + \sum_{r=1}^{s}\sum_{x \in \delta^{-1}(x_r^i)} \sum_{t\in[2]} f_i(x\,|\,S_{it}) 
            \le 2f_i(S_{i}) + \sum_{r=1}^{s}\sum_{x \in \delta^{-1}(x_r^i)} \sum_{t\in[2]} f_i(x\,|\,S^{(r)}_{it}) \\
            &\le 2f_i(S_{i}) + \sum_{r=1}^{s} 2 (n + p_i ) f_i\big(x_r^i \,|\, S_{i,\tau(r)}^{(r)}\big)  
            = 2f_i(S_{i}) + 2 (n + p_i ) (f_i(S_{i1}) +f_i(S_{i2}))  \\[5pt]
            &\le 2f_i(S_i) + 4(n + p_i) f_i(S_i) 
            = (4n + 4p_i + 2)  f_i(S_i) \,.
\end{align*}
where the second inequality follows by the definition of $Z$, the third inequality follows from observing that $\bigcup_{r=1}^{s} \delta^{-1}(x_r^i) = O_i^-\setminus Z$, the fourth follows from submodularity, and the fifth from applying both parts of Lemma \ref{lem:mapping-non-monotone-general}.
We conclude that $f_i(S_i) \ge \sum_{t\in[2]} f_i(O_i^-\cup S_{it}) / (4n + 4p_i + 2)$.

We still need to relate $\sum_{t\in[2]} f_i(O_i^-\cup S_{it})$ to $f_i(O_i^-)$.
For this, we are going to need the following classic lemma that connects sampling and 
submodular maximization \citep{FeigeMV11,BuchbinderFNS14}. 
\begin{lemma}[Lemma 2.2 of \citet{BuchbinderFNS14}]
	\label{lem:sampling}
	Let $g: 2^M \to \mathbb{R}$ be a (possibly not normalized) submodular set function, let $ X \subseteq 
	M$ and let $X(p)$ be a sampled subset, where each element of $X$ appears 
	with probability at most $p$ (not necessarily independent). Then
	$\mathbb{E} \left[ g(X(p)) \right] \geq (1-p)g(\emptyset)$.
\end{lemma}

We are going to follow the ``derandomization'' of the lemma as done by \citet{FeldmanHK23}. 
We begin by setting $g(S) = f_i(S\cup O_i^-)$, whereas our random set $X(p)$ is going to be one of $S_{i1}$ and $S_{i2}$, uniformly at random. Clearly, $p=0.5$ here and, therefore, applying Lemma \ref{lem:sampling} gives us the desired approximation
\begin{align*}
\optm_i &\le  f_i(O_i^-) =  g(\emptyset) 
            \le 2\, \mathbb{E} \left[ g(X(p)) \right] = \sum_{t\in[2]} f_i(O_i^-\cup S_{it})
            \le  (4n + 4p_i + 2) f_i(S_i)\,.\qedhere
\end{align*}
\end{proof}


\begin{proof}[\textbf{Proof of Lemma \ref{lem:mapping-non-monotone-general}}]
Similar to what we did in the proof of Theorem \ref{thm:greedy_monotone_general}, we are going to define a partition of $P$. First, let 
$M_i^{(r)} = \{x_r^i, x_r^{i+1}, \ldots, x_r^{i-1}, x_{r+1}^i, \ldots\}$ be the set of available items right before $x^i_r$ is added to $S_{i,\tau(r)}^{(r)}$, where $\tau(r)\in \{1,2\}$ denotes the index of the solution which was augmented with $x^i_r$ for any $r\le s$ (recall that $s= |S_{i1}| +|S_{i2}|$). Now, for $r \in [s]$, let
\begin{align*}
X^{t}_{(r)} =  \big\{x\in P\cap M_i^{(r)} :\, S_{it}^{(r)}\cup\{x\} \in \mathcal{I}_i 
          \text{\ \ and\ \ } S_{it}^{(r)}\cup\{x^i_r, x\} \notin \mathcal{I}_i \big\}\,,
\end{align*}
if $t = \tau(r)$, and $X^{t}_{(r)} = \emptyset$, if $t\in \{1,2\}\setminus\{\tau(r)\}$.
That is, $X^{\tau(r)}_{(r)}$ contains all items of $P$ that were available and could be added to $S_{i,\tau(r)}^{(r)}$ right before $x^i_r$ was added, but become infeasible right after. Also, for $r \in [s]$ (using the convention that $M_i^{(|S_i|+1)} = \emptyset$), let 
\begin{align*}
Y_{(r)} =  \big\{x\in \big(P\cap M_i^{(r)}\big) \setminus \big(X_{(r)} \cup M_i^{(r+1)} \big):\, 
          S_{it}^{(r)}\cup\{x\} \in \mathcal{I}_i \text{\ \ for both\ \ } t\in\{1, 2\} \big\}\,, 
\end{align*}
i.e., $Y_{(r)}$ contains all the items of $P \setminus X_{(r)}$ that were available and could be added to \emph{both} solution sets right \emph{before} $x^i_r$ was added but were picked before agent $i$ got to choose for the $(r+1)$-th time. Note that $Y_{(r)}$ could contain $x^i_r$ itself.
It is easy to see that 
\[P = \textstyle \bigcup\limits_{r=1}^{s}(X^1_{(r)} \cup X^2_{(r)} \cup Y_{(r)}) \,.\]
Similarly to the proof of Lemma \ref{lem:mapping-monotone-general}, it is clear that $|Y_{(r)}|\le n$ for all  $r< s$ (since now $x_r^i$ may also be included in $Y_{(r)}$). 
For $r = s$ note that
$|Y_{(r)}|\ge n+1$ would mean that agent $i$ still has items that they could add to at least one of $S_{i1}$ or $S_{i2}$ in the next round of the protocol, contradicting the fact that $s$ is the sum of the cardinalities of $i$'s solutions.

We further claim that $|\bigcup_{j=1}^{r} X^t_{(j)}| \le a p_i|S_{it}^{(r+1)}| $ for all $r \in [s]$ and $t\in\{1,2\}$. Suppose not, and let $k$ be the smallest index for which $|\bigcup_{j=1}^{k} X^t_{(j)}| > a p_i|S_{it}^{(k+1)}|$ for some $t\in\{1,2\}$. 
By the the pigeonhole principle, there exists some $\ell\in[a]$, such that $|\bigcup_{j=1}^{k} (X^t_{(j)}\cap Q_{i\ell})| > p_i|S_{it}^{(k+1)}|$.
Also, because of $k$'s minimality, it must be the case that $S_{it}^{(k+1)} = S_{it}^{(k)}\cup x^i_k$ and, of course, $S_{it}^{(k+1)}\in \mathcal{I}_i$.
Moreover, by the definition of the $X^t_{(r)}$'s, $S_{it}^{(k+1)}$ must be maximally independent in the set $S_{it}^{(k+1)} \cup \bigcup_{j=1}^{k} (X^t_{(j)} \cap Q_{i\ell})$. 
Indeed, each $x \in \bigcup_{j=1}^{k} (X^t_{(j)} \cap Q_{i\ell})$ belongs to some $X^t_{(j)}$ for $j\in[k]$ and, thus, 
$S_{it}^{(j)}\cup\{x^i_j, x\} \notin \mathcal{I}_i$. Since $S_{it}^{(j)}\cup\{x^i_j\} \subseteq S_{it}^{(k+1)}$, we also have that $S_{it}^{(k+1)}\cup\{x\} \notin \mathcal{I}_i$. This implies that 
\[\textrm{lr}\big(S_{it}^{(k+1)} \cup \textstyle \bigcup\limits_{j=1}^{k} (X^t_{(j)} \cap Q_{i\ell})\big) \le |S_{it}^{(k+1)}|\,.\] 
On the other hand, $\bigcup_{j=1}^{k} (X^t_{(j)} \cap Q_{i\ell}) \in \mathcal{I}_i$ because $\bigcup_{j=1}^{k} (X^t_{(j)} \cap Q_{i\ell}) \subseteq Q_{i\ell} \in \mathcal{I}_i$. As a result 
\[\textrm{ur}\big(S_{it}^{(k+1)} \cup \textstyle \bigcup\limits_{j=1}^{k} (X^t_{(j)} \cap Q_{i\ell})\big) \ge \Big|\textstyle \bigcup\limits_{j=1}^{k} (X^t_{(j)} \cap Q_{i\ell})\Big|> p_i |S_{it}^{(k+1)}| \,,\]
contradicting the fact that $\mathcal{I}_i$ is a $p_i$-system for $M$. We conclude that $|\bigcup_{j=1}^{r} X^t_{(j)}| \le a p_i|S_{it}^{(r+1)}| $ for all $r \in [s]$ and $t\in\{1,2\}$.

Now, we can define $\delta$. Recall that $\tau(r), r\in[s]$, denotes the index of the solution augmented with $x^i_r$. For $t\in\{1,2\}$, we rename the elements of $\bigcup_{j=1}^{|S_i|} X^t_{(j)}$ so that all items in $X^t_{(j)}$ are lexicographically before any item in $X^t_{(j')}$ for $j<j'$. 
We map to $x_1^i$ the first $a p_i$ items of $\bigcup_{j=1}^{|S_i|} X^{(\tau(1))}_{(j)}$ (with respect to the aforementioned ordering) as well as all the items in $Y_{(1)}$, we then map to $x_2^i$ the first $a p_i$ \emph{unmapped} items of $\bigcup_{j=1}^{|S_i|} X^{(\tau(2))}_{(j)}$ as well as all the items in $Y_{(2)}$, and so forth, until the items of $P$ are exhausted. Since at most $p_i+n$ items are mapped to each element of $S_{i1} \cup S_{i2}$, the second claimed property of $\delta$ directly holds.

For the first property, we start with any $x\in Y_{(r)}$ for $r \in [s]$. Because $x$ was available and feasible to both $S_{i1}^{(r)}, S_{i2}^{(r)}$ when $x_r^i$ was chosen, we have $f_i(x \,|\, S_{it}^{(r)}) \le f_i(x_r^i \,|\, S_{i,\tau(r)}^{(r)})$ as desired, given that $\delta(x) = x_r^i$ in this case.
Next, consider any $x\in P\setminus \bigcup_{r=1}^{s}Y_{(r)}$ and notice that it either belongs to a single set $X^t_{(r)}$, $t\in\{1,2\}$, or to two sets $X^1_{(j)}$ and $X^2_{(j')}$; in the latter case, let $r = \min\{j, j'\}$. The way we define $\delta$, and because of the upper bound we derived on $|\bigcup_{j=1}^{r} X^t_{(j)}|$, it is guaranteed that $x$ is mapped to some $x_{\kappa}^i$ with $\kappa\le r$. This means that $x$ was available and feasible to both $S_{i1}^{(\kappa)}, S_{i2}^{(\kappa)}$ when $x_{\kappa}^i$ was chosen, thus $f_i(x \,|\, S_{it}^{(\kappa)}) \le f_i(x_r^i \,|\, S_{i,\tau(\kappa)}^{(\kappa)})$ as desired.
\end{proof}

Specifically for cardinality constraints, we can get a somewhat stronger version of Theorem \ref{thm:greedy_non-monotone_general}, like we did in the monotone case.  

\begin{theorem}\label{thm:greedy_non-monotone_cardinality}
Any agent $i$ with a non-monotone objective and a cardinality constraint, 
who chooses greedily in 
the Round-Robin protocol, builds two solutions $S_{i1}, S_{i2}$ such that $\max_{t\in\{1, 2\}} f_i(S_{it}) \ge {\optm_i}/{(4n +2)}$. 
\end{theorem}

\begin{proof}
The proof closely follows that of Theorem \ref{thm:greedy_non-monotone_general}, so we are just highlighting the differences here.
Consider an agent $i$ with a cardinality constraint of $k$. The second property of the mapping $\delta$ now is slightly stronger. In particular, $|\delta^{-1}(x_r^i)| \le n $ (rather than having $n+1$ on the right-hand side as expected from the proof of Theorem \ref{thm:greedy_non-monotone_general}). This means that after we apply Theorem \ref{thm:SM-general} for $O_i^-\cup S_{it}$ and $S_{it}$, and we sum up for $t\in\{1, 2\}$, we have $\sum_{t\in[2]} f_i(O_i^- \cup S_{it}) \le (4n + 2)  f_i(S_i) $ and, through Lemma \ref{lem:sampling}, we finally get $\optm_i \le (4n + 2) f_i(S_i)$ the same way that we got the approximation guarantee in the proof of Theorem \ref{thm:greedy_non-monotone_general}. The only thing remaining is the existence of $\delta$. The main idea is what we did for the mapping in Lemma \ref{lem:mapping-non-monotone-general}, 
but, like in Theorem \ref{thm:greedy_monotone_cardinality}, we only do it for the case of a single set. The reason is that this mapping is only applied here for $a = 1$, $Q_{i1} = O_i^-$ and $P = O_i^-\setminus Z$, and we do not need it for general $a$; i.e., $\delta: O_i^-\setminus Z \to S_{i1} \cup S_{i2}$.  Nevertheless, such a $\delta$ exists for any $a\le n$.

Let $s \vcentcolon= |S_{i1}| +|S_{i2}|$.
For $r \in [s]$, we define $X^1_{(r)}$, $X^2_{(r)}$, and $Y_{(r)}$ like before, and notice that if $\max\{|S^{(r)}_{i1}|, |S^{(r)}_{i2}|\} < k$, then $X^1_{(r)} =X^2_{(r)} = \emptyset$. 

If $\bigcup_{r=1}^{s}Y_{(r)} = O_i^-\setminus Z$, then $\delta$ just maps all items of $Y_{(r)}$ to $x_r^i$ for all $r \in [s]$.
Next, suppose $|(O_i^-\setminus Z)\setminus \bigcup_{r=1}^{s}Y_{(r)}|= k' \in [k]$. Clearly, this means that $\bigcup_{r=1}^{s} (X^1_{(r)} \cup X^2_{(r)}) \neq \emptyset$ and thus, $s \ge \max\{|S_{i1}|, |S_{i2}|\} \ge k$. 
Just by counting, we have that there are at least $k'$ values of $r$ in $[k]$, such that $Y_{(r)} = \emptyset$; call $r_1, r_2, \ldots, r_{k'}$ the smallest $k'$ such values, and let $z_1, z_2, \ldots, z_{k'}$ be the items in $(O_i^-\setminus Z)\setminus \bigcup_{r=1}^{s}Y_{(r)}$. Now, set $\delta(z_j) = x_{r_j}^i$, for $j\in[k]$, whereas if $x\in Y_{(r)}$ then $\delta(x) = x_r^i$  for all $r\in[k]$. Clearly, in both cases, we never map more than $n$ items to a single element of $S_{i1} \cup S_{i2}$, which proves the second desired property of $\delta$.

The first property, both in the case of some $x\in Y_{(r)}$, $r \in [s]$, and in the case of an $x\in (O_i^-\setminus Z)\setminus \bigcup_{r=1}^{s}Y_{(r)}$  follows exactly like in the proof of Lemma \ref{lem:mapping-non-monotone-general}.
\end{proof}

\subsection{Implications to Constrained Fair Division}
\label{sec:fairness-non-monotone}
Surprisingly---given the scarcity of similar results for non-monotone objectives---we can get the analogs of Theorems \ref{thm:greedy_monotone_ef1} and \ref{thm:greedy_monotone_cardinality_ef1}. Note, however, that here the set $R_j$ in the statements of Theorems \ref{thm:greedy_nonmonotone_ef1} and \ref{thm:greedy_nonmonotone_cardinality_ef1} below is the set of \textit{all} items that agent $j$ received and not necessarily a feasible set for $j$.  For instance, if $j$ chooses greedily, then $R_j = S_{j1} \cup S_{j2}$.
This results in weaker results from a fair division point of view; see Corollaries \ref{cor:fairness_p-system_nm} and \ref{cor:fairness_cardinality_nm} below.

\begin{theorem}\label{thm:greedy_nonmonotone_ef1}
Let $i$ be an agent with a non-monotone objective and a $p_i$-system constraint $\mathcal{I}_i$ who chooses greedily in the Round-Robin protocol. Also, let $j$ be any agent, $R_j$ be the set of all items $j$ got, and $g$ be the first item added to $R_j$. Then, 
$\max_{t\in\{1, 2\}} f_i(S_{it})
\ge \max_{S\in \mathcal{I}_i|R'_j}f_i(S) / (4p_i + 10)$, where $R'_j = R_j$ if $i<j$, and $R'_j = R_j\setminus \{g\}$ otherwise.
Moreover, if $M' = M\setminus \cup_{j\in [n]} R_j$, then $\max_{t\in\{1, 2\}} f_i(S_{it}) \ge \max_{S\in \mathcal{I}_i|M'}f_i(S) / (4p_i + 6)$.
\end{theorem}

\begin{proof}
The proof is similar to the proof of Theorem \ref{thm:greedy_monotone_ef1}. Now both parts of the statement are corollaries of Theorem \ref{thm:greedy_non-monotone_general}. 

Let $j$ be any agent other than $i$ and $R'_j$ be as given in the statement. We can view $R_i = S_{i1} \cup S_{i2}$ and $R_j$ as the result 
of running the Round-Robin protocol on $R_i \cup R_j$ with just $i$ and $j$ following the same order and the same choices as before. 
From the point of view of $i$, all choices are consistent with the greedy policy $\mathcal{G}^+_i$. Thus, by 
Theorem \ref{thm:greedy_non-monotone_general}, 
\[\max_{t\in\{1, 2\}} f_i(S_{it})  \ge \max_{S\in \mathcal{I}_i|R_i \cup R_j}f_i(S) / (4\cdot 2+ 4p_i + 2) \ge \max_{S\in \mathcal{I}_i|R'_j}f_i(S) / (4p_i +10) \,,\] 
where the last inequality follows from the fact that $R'_j \subseteq R_i \cup R_j$.

Next, we assume the existence of an auxiliary dummy agent $i'$,
with the  trivial independence system constraint $\mathcal{I}_{i'} = \{\emptyset\}$. Then, we view $R_i$ and $M'$ as the result 
of running the Round-Robin protocol on $R_i \cup M'$ with just  $i$ and $i'$ following  the same choices for $i$ as before. 
From $i$'s point of view, these choices are consistent with the policy $\mathcal{G}^+_i$. Now, however, the sets $Y_{(r)}$ in the proof
of Lemma \ref{lem:mapping-non-monotone-general} are all empty \textit{or singletons} (as $i'$ never picks any items but $Y_{(r)}$ may contain $x^i_r$) rather than having cardinality up to $2$. 
Thus, part \textit{2.}~of Lemma \ref{lem:mapping-non-monotone-general}, for $n = 2$ and $a=1$, becomes $|\delta^{-1}(x_r^i)| \le  p_i +1$ instead of $|\delta^{-1}(x_r^i)| \le p_i + 2$, improving the overall factor of Theorem \ref{thm:greedy_non-monotone_general} by $4$ in this case. 
That is, by Theorem \ref{thm:greedy_non-monotone_general},  
\[\max_{t\in\{1, 2\}} f_i(S_{it})  \ge \max_{S\in \mathcal{I}_i|R_i \cup M'}f_i(S) / (4 + 4p_i + 2) \ge \max_{S\in \mathcal{I}_i|M'}f_i(S) / (4p_i +6) \,,\] 
where again the last inequality is due to $M' \subseteq R_i \cup M'$.
\end{proof}

\begin{theorem}\label{thm:greedy_nonmonotone_cardinality_ef1}
	Let $i$ be an agent with a non-monotone objective and a cardinality constraint who chooses greedily in the Round-Robin protocol. Also, let $j$ be any agent, $R_j$ be the set of all items $j$ got, and $g$ be the first item added to $R_j$. Then, 
	$\max_{t\in\{1, 2\}} f_i(S_{it})
	\ge \max_{S\in \mathcal{I}_i|R'_j}f_i(S) / 10$, where $R'_j = R_j$ if $i<j$, and $R'_j = R_j\setminus \{g\}$ otherwise.
\end{theorem}

\begin{proof}
	The proof closely follows the first parts of the proofs of Theorems \ref{thm:greedy_monotone_ef1} and \ref{thm:greedy_nonmonotone_ef1}.
	The setup and the notation are the same, and the only difference is the following: when we compare the best set among $S_{i1}$ and $S_{i2}$ to the best feasible subset of $R_i \cup R_j$ from $i$'s point of view,  we invoke Theorem \ref{thm:greedy_non-monotone_cardinality} (for $n=2$) and get 
\[\max_{t\in\{1, 2\}} f_i(S_{it})  \ge \max_{S\in \mathcal{I}_i|R_i \cup R_j}f_i(S) / (4\cdot 2+ 2) \ge \max_{S\in \mathcal{I}_i|R'_j}f_i(S) / 10 \,,\] 
	where  the second inequality is due to $R'_j \subseteq R_i \cup R_j$.
\end{proof}

One would hope that by simulating a run of Round-Robin$(\bm{\mathcal{G}^+})$, i.e., our Protocol \ref{alg:MRR} with greedy policies for all agents, we would get the counterparts of Corollaries \ref{cor:fairness_p-system} and \ref{cor:fairness_cardinality}. However, the corresponding results here are with respect to the functions $\hat{f}_i(T):= \max_{S\in \mathcal{I}_i|T}f_i(S)$, for $i\in [n]$, rather than the original $f_i$s. An alternative way to see these results is to weaken the guarantee towards the unallocated items, by ``throwing away'' the least favorite among $S_{i1}$ and $S_{i2}$ for each agent $i$.

\begin{corollary}\label{cor:fairness_p-system_nm}
    For agents with non-monotone submodular valuation functions and $p$-system constraints, we can efficiently find a $\frac{1}{4p+10}$-FEF1 and $\frac{1}{4p+6}$-FEFu maximal allocation with respect to the functions $\hat{f}_i$, $i\in [n]$. This implies a $\frac{1}{4p+10}$-FEF1 and $\frac{1}{(n+1)(4p+10)}$-FEFu allocation with respect to the original objective functions.
\end{corollary}

\begin{corollary}\label{cor:fairness_cardinality_nm}
    For agents with non-monotone submodular valuation functions and cardinality constraints, we can efficiently find a $0.1$-FEF1 and $0.1$-FEFu maximal allocation with respect to the functions $\hat{f}_i$, $i\in [n]$. This implies a $\frac{1}{10}$-FEF1 and $\frac{1}{10(n+1)}$-FEFu allocation with respect to the original objective functions.
\end{corollary}

The $(n+1)$ factor in the FEFu guarantees for the original functions in Corollaries \ref{cor:fairness_p-system_nm} and \ref{cor:fairness_cardinality_nm} arises because when each agent $i$ discards the weaker of $S_{i1}, S_{i2}$, the  pool of unallocated items grows by the discarded sets (which are $n$ in total).

\subsection{Improved Guarantees for Robust Instances}
\label{sec:robust_non-monotone}
Like in Section \ref{sec:robust_monotone} for the monotone case, here we explore what is possible for the special case of robust instances. We show that the linear factors of Theorems \ref{thm:greedy_non-monotone_general} and \ref{thm:greedy_non-monotone_cardinality} can still be removed for instances that are $(\Omega(n), O(1))$-robust with respect to everyone. We begin with the analog of Theorem \ref{thm:greedy_monotone_general_robust}.

\begin{theorem}\label{thm:greedy_non-monotone_general_robust}
    Assume that right before agent $i$ is allowed to pick their first item, the remaining instance 
    is $(\gamma n, \beta)$-robust with respect to $i$ (who has a non-monotone objective and a $p_i$-system constraint). 
    By choosing greedily in the Round-Robin protocol, $i$  
    builds two solutions $S_{i1}, S_{i2}$ such that 
    $\max_{t\in\{1, 2\}} f_i(S_{it}) \ge {\opt_i}/{ 2\beta( 2p_i + 1 +2/\gamma)}$.
\end{theorem}

\begin{proof}
Our notation will be consistent with the one introduced in the proof of Theorem \ref{thm:greedy_non-monotone_general}. 
In particular, the set $M_i$ contains the items still available right before agent $i$ chooses their very first item, 
$x_j^{\ell}$ is the $j$-th item that agent $\ell$ chooses from the moment when agent $i$ is about to start choosing, 
$S_{it}^{(r)}, t\in\{1,2\}$ are the solutions of agent $i$ right before item $x_r^i$ is added to one of them, 
$\tau(r)$ denotes the index of the solution augmented with $x^i_r$ for any $r\le s \vcentcolon= |S_{i1}| +|S_{i2}|$, 
and $S_i\vcentcolon= S_{i\tau}$ with $\tau \in \argmax_{t\in[2]} f_i(S_{it})$.

Similarly to the proof of Theorem \ref{thm:greedy_monotone_general_robust}, since the instance is 
$(\gamma n, \beta)$-robust with respect to $i$ right before $i$ gets to pick their first item, 
at that time there are at least $a \vcentcolon= \gamma n$ disjoint subsets of $M_i$, say $Q_{i1}, \ldots, Q_{i a}$, 
such that $Q_{ij}\in \mathcal{I}_i$ and $\beta \sum_{j\in[a]} f_i(Q_{ij}) \ge a\,\opt_i$.

We assume the existence of a mapping $\delta:  P \to S_{i1} \cup S_{i2}$ as in the statement of Lemma \ref{lem:mapping-non-monotone-general}; recall that $P = \{x\in \bigcup_{j\in[a]} Q_{ij} :\, \text{\ there is\ } t \in \{1, 2\} \text{\ such that\ } f_i(x\,|\, S_{it}) \ge 0\}$.  Further, we apply the first part of Theorem \ref{thm:SM-monotone} for each pair $Q_{ij} \cup S_{it}$ and $S_{it}$, for $j\in [a]$ and $t\in\{1,2\}$; we add up these $2a$ inequalities: 
\begin{align*}
 \sum_{j=1}^{a} \sum_{t=1}^{2}  f_i(Q_{ij} \cup S_{it})  &\le \sum_{j=1}^{a}\sum_{t=1}^{2}  f_i(S_{it}) + \sum_{j=1}^{a}\sum_{t=1}^{2} \sum_{x \in Q_{ij} \setminus S_{it}}  f_i(x\,|\,S_{it})   \\
            & = 2a  f_i(S_i) + \sum_{j=1}^{a} \sum_{t=1}^{2} \sum_{x \in Q_{ij}}  f_i(x\,|\,S_{it}) \\ 
            & \le 2a  f_i(S_i) + \sum_{t=1}^{2} \sum_{x \in P}  f_i(x\,|\,S_{it}) \\ 
            &\le 2a f_i(S_{i}) + \sum_{t=1}^{2} \sum_{r=1}^{s}\sum_{x \in \delta^{-1}(x_r^i)} \!\! f_i(x\,|\,S_{it}) \\
            & \le 2a f_i(S_{i}) +  \sum_{r=1}^{s}\sum_{x \in \delta^{-1}(x_r^i)} \sum_{t=1}^{2} f_i(x\,|\,S^{(r)}_{it}) \\
            &\le 2a f_i(S_{i}) + \sum_{r=1}^{s} 2 (n + a p_i ) f_i\big(x_r^i \,|\, S_{i,\tau(r)}^{(r)}\big) \\ 
            & = 2a f_i(S_{i}) + 2 (n + a p_i ) (f_i(S_{i1}) +f_i(S_{i2}))  \\[5pt]
            &\le 2 \gamma n f_i(S_i) + 4(n + \gamma n p_i) f_i(S_i) \\[5pt]
            &\le 2n (\gamma + 2+ 2\gamma p_i)  f_i(S_i) \,,
\end{align*}
where for the first equality we potentially just added some terms of zero value to the triple sum,
the second inequality follows by the definition of $P$ (namely, we possibly removed some negative summands), the third inequality follows from recalling that $\bigcup_{r=1}^{s} \delta^{-1}(x_r^i) = P$, the fourth follows from submodularity, and the fifth from applying both parts of Lemma \ref{lem:mapping-non-monotone-general}.
We conclude that 
\[f_i(S_i) \ge \sum_{j=1}^{a} \sum_{t=1}^{2} f_i(Q_{ij} \cup S_{it}) / 2n (\gamma + 2+ 2\gamma p_i)\,.\]

We still need to relate $\sum_{j\in[a]}\sum_{t\in[2]} f_i(Q_{ij} \cup S_{it})$ to $\opt_i$ via Lemma \ref{lem:sampling}.

For this, we are going to apply the lemma for each $g_j(S) = f_i(S\cup Q_{ij})$, where the random set $X(p)$ is one of $S_{i1}$ and $S_{i2}$, uniformly at random, i.e., $p=0.5$. Thus, Lemma \ref{lem:sampling} gives 
\begin{equation*}
f_i(Q_{ij}) =  g_j(\emptyset) 
            \le 2  \mathbb{E} \left[ g_j(X(p)) \right] =  \sum_{t=1}^{2} f_i(Q_{ij} \cup S_{it})\,.
\end{equation*}
By summing up for all values of $j$ and using $(\gamma n, \beta)$-robustness (recall $a = \gamma n$), we get
\begin{equation*}
\gamma n\,\opt_i  \le \beta \sum_{j=1}^{a} f_i(Q_{ij}) \le   \beta \sum_{j=1}^{a} \sum_{t=1}^{2} f_i(Q_{ij} \cup S_{it}) \le  2n \beta (\gamma + 2+ 2\gamma p_i) f_i(S_i)\,,
\end{equation*}
and we conclude that $ \max_{t\in\{1, 2\}} f_i(S_{it}) = f_i(S_i) \ge {\opt_i}/{ 2\beta( 2p_i + 1 +2/\gamma)}$. 
\end{proof}

The next corollary, follows from Theorem \ref{thm:greedy_non-monotone_general_robust} in exactly the same way that 
Corollary \ref{cor:greedy_monotone_general_robust_1} follows from Theorem \ref{thm:greedy_monotone_general_robust}.

\begin{corollary}\label{cor:greedy_non-monotone_general_robust_1}
    Assume that an instance is $(\gamma n +i -1, \beta)$-robust with respect to agent $i$ (who has a non-monotone objective and a $p_i$-system constraint). By choosing greedily in the Round-Robin protocol, $i$  
    builds two solutions $S_{i1}, S_{i2}$ such that 
    $\max_{t\in\{1, 2\}} f_i(S_{it}) \ge {\opt_i}/{ 2\beta( 2p_i + 1 +2/\gamma)}$.
\end{corollary}

Through Corollary \ref{cor:greedy_non-monotone_general_robust_1} 
we obtain strong guarantees with respect to all agents for universally robust instances. 

\begin{corollary}\label{cor:greedy_non-monotone_general_robust_2}
    Assume that an instance is $(\lceil(1+\gamma) n\rceil , \beta)$-robust with respect to every agent for constant $\beta, \gamma \in \mathbb{R}_+$. Any agent $i$ with a non-monotone objective and a $p_i$-system constraint 
    who chooses greedily in
    the Round-Robin protocol, achieves a $1/\Theta(p_i)$ fraction of\, $\opt_i$. In particular, for $\beta = 1, \gamma  = 2$, any agent $i$ achieves $\opt_i/(4p_i + 4)$ for a $p_i$-system constraint (hence, $\opt_i/8$ for a cardinality constraint).
\end{corollary}

\begin{proof}
Like in the proof of Corollary \ref{cor:greedy_monotone_general_robust_2},  
$\lceil(1+\gamma) n\rceil  = n + n\zeta(n) $, with $\gamma - 1/n < \zeta(n) \le \gamma$.
If an instance is $((1+\zeta(n)) n , \beta)$-robust with respect to agent $i$, it is 
also $(\zeta(n) n +i -1, \beta)$-robust with respect to $i$. 
By Corollary \ref{cor:greedy_non-monotone_general_robust_1}, we get 
$\opt_i \le  2\beta( 2p_i + 1 +2/\zeta(n)) f_i(S_i)$, and since $\beta$ and 
$\gamma$ are constant, $2\beta( 2p_i + 1 +2/\zeta(n)) = \Theta(p_i)$. 
In particular, $2\beta( 2p_i + 1 +2/\zeta(n)) = 4p_i + 4$, for $\beta = 1, \gamma  = 2$. 
\end{proof}

\section{Maximin Share Guarantees via High-Valued Singletons}
\label{sec:mms}
As we mentioned earlier, it is possible to obtain MMS guarantees by adding a short extra phase to Protocol \ref{alg:MRR}. This is a standard trick done 
to obtain approximate MMS allocations \citep{AMNS17,BarmanK20},  although here there are some additional complications. 
In particular, the Augmented Round-Robin protocol (Protocol \ref{alg:AugRR}) initially gives each agent the option to pick a single high-value item in Phase~1 and exit, before running Protocol~\ref{alg:MRR} on the remaining agents and items in Phase~2.

\makeatletter
\renewcommand{\ALG@name}{Protocol}
\makeatother

\setcounter{algorithm}{1}

\begin{algorithm}[ht]
		\caption{Augmented Round-Robin$({\mathcal{A}}_1, \ldots, {\mathcal{A}}_n)$ \\{\small {(For $i\in [n]$, ${\mathcal{A}}_i$ is the policy of   agent $i$.)}}}
		\begin{algorithmic}[1]
			\State $Q=M$\textbf{;} $R = [n]$
			\For{$i = 1, \dots, n$}   {\small\hfill \textbf{Phase 1}}
				\State $j = \mathcal{A}_i(\emptyset\,; Q\,;1)$ {\small\hfill (where $j$ could be a \textit{dummy} item)} 
				\If{$j\in Q$}{\small\hfill i.e., if $j$ is not a dummy item} 
					\State $Q = Q\setminus \{j\}$ {\small\hfill update remaining items}
					\State $R = R\setminus \{i\}$ {\small\hfill update remaining agents}
				\EndIf
			\EndFor
			\State Run Protocol \ref{alg:MRR} on the instance induced by $R$ and $Q$. {\small\hfill \textbf{Phase 2}}
		\end{algorithmic}
		\label{alg:AugRR}
\end{algorithm}

\makeatletter
\renewcommand{\ALG@name}{Policy}
\makeatother

\setcounter{algorithm}{2}

\begin{algorithm}[ht]
		\caption{ Augmented Greedy policy $\mathcal{H}_i(S_i\,; Q\,; b)$ of agent $i$ (assuming access to an MMS oracle). \\ {\small {($S_i$\,: current solution of agent $i$ (initially $S_{i} =\emptyset$)\,; $Q$\,: the current set of available items\,; $b\in \{1, 2\}$\,: phase number)}}}
		\begin{algorithmic}[1]
			\vspace{2pt}\State $\theta = {\small \begin{cases}
                3\ , & \text{for a monotone objective and a cardinality constraint}  \\
                p_i + 3\ , & \text{for a monotone objective and other $p_i$-system constraints}  \\
                4p_i + 7\ , & \text{for a non-monotone objective and a $p_i$-system constraint} 
                \end{cases}}$ \label{line:theta} \vspace{2pt}
            \State $A =\{x\in Q\,:\, S_i \cup \{x\} \in \mathcal{I}_i\}$
            \If{$b = 1$} 
                \If{$A \neq \emptyset$ \textbf{and} $\max_{z\in A} f_i(z\,|\,S_i) \ge \fmms_i / \theta$ \label{line:mms_single}}
                    \State $j\in \argmax_{z\in A} f_i(z\,|\,S_i)$ 
                    \State $S_i = S_{i} \cup \{j\}$ \label{line:mms_j}
                \Else
                    \State $j$ is a dummy item
                \EndIf
            \State \Return{$j$}
			\Else
			\State Follow Greedy policy $\mathcal{G}_i(S_i\,; Q)$ (or $\mathcal{G}^+_i(S_{i1}, S_{i2}\,; Q)$ if $f_i$ is non-monotone)
			\EndIf
		\end{algorithmic}
		\label{alg:aug_greedy_policy}
\end{algorithm}\medskip

\begin{theorem}\label{thm:greedy_monotone_general_MMS-oracle}
Any agent $i$ with a monotone objective and a $p_i$-system (resp.~cardinality) constraint, who has access to an $\fmms$ oracle and 
chooses according to the augmented greedy policy $\mathcal{H}_i$ in the Augmented Round-Robin protocol, 
builds a solution $S_i$ such that $f_i(S_i) \ge {\fmms_i}/{(p_i + 3)}$ (resp.~$f_i(S_i) \ge {\fmms_i}/{3}$).
\end{theorem}

\begin{proof}
It suffices to show that, in any case, $f_i(S_i) \ge {\fmms_i}/{\theta}$, where $\theta$ is equal to $p_i + 3$ for general $p_i$-system constraints but only $3$ for cardinality constraints (see line \ref{line:theta} of Policy \ref{alg:aug_greedy_policy}).

If agent $i$ gets an item $j$ in Phase $1$, then $i$ is not allowed to get another item, i.e., $S_i = \{j\}$.
However, in this case, by lines \ref{line:mms_single}-\ref{line:mms_j} of policy $\mathcal{H}_i$, we directly have that 
$f_i(S_i) = \max_{z\in A} f_i(z\,|\,\emptyset) \ge \fmms_i / \theta$.

Now suppose that agent $i$ did not get an item in Phase $1$ but rather participates in Protocol \ref{alg:MRR} that runs on the 
instance induced by $R$ and $Q$ in Phase $2$. If this happened because $A=\emptyset$ in line \ref{line:mms_single}, then $\fmms_i = 0$ and the 
statement trivially holds for $S_i = \emptyset$. So, assume that this is not the case. Let $\opt'_i:= \max_{S\in \mathcal{I}_i|Q}f_i(S)$ and $r:=|R|$. 
We first observe that $\fmms_i\le \fmms_i(r, Q)$, just by applying 
Lemma \ref{lem:monotonicity} $n-r$ times.

Next, we claim that right before agent $i$ is allowed to pick their first item 
in Phase 2, i.e., in Protocol \ref{alg:MRR} on $R$ and $Q$, the remaining instance is $\big(r, \frac{\theta\, \opt'_i}{(\theta-1)\fmms_i(r, Q)}\big)$-robust 
with respect to $i$. Equivalently, we will show that at that moment there are $r$ sets $O_1, O_2, \ldots, O_{r}$ of average 
value $(\theta-1)\fmms_i(r, Q) / \theta$. To this end, consider an $\fmms$-defining $r$-partition of $Q$ with respect to $i$, i.e., a partition 
$(Q_1, Q_2, \ldots, Q_{r})$ of $Q$, such that $f_i(Q_t) \ge \fmms_i(r, Q)$, for any $t\in [r]$. 
At the time right before agent $i$ is allowed to pick their first item, at most $r$ items of $Q$ have been allocated to agents in $R\cap [i-1]$,
say items $x_1, \ldots, x_{\ell}$, where $\ell \le r$. Since the condition of line \ref{line:mms_single} was false, $f_i(x_t) < \fmms_i / \theta$,
for any $t\in [\ell]$. If we define $O_t := Q_t \setminus \{x_1, \ldots, x_{\ell}\}$, for $t\in [r]$, then the total value of these sets is
\[\sum_{t=1}^{r} f_i(O_t) = \sum_{t=1}^{r} f_i(Q_t) - \sum_{t=1}^{\ell} f_i(x_t) \ge r\cdot \fmms_i(r,Q) - \ell\cdot\frac{\fmms_i}{\theta} \ge r\cdot \frac{\theta-1}{\theta}\,\fmms_i(r,Q)\,, \]
and, hence, their average value is $(\theta-1)\fmms_i(r, Q) / \theta$, as claimed.

Finally, we can apply Theorem \ref{thm:greedy_monotone_general_robust} with $\gamma = 1$ and $\beta = \frac{\theta\, \opt'_i}{(\theta-1)\fmms_i(r, Q)}$ for agent $i$ on the execution of Protocol \ref{alg:MRR} on the 
instance induced by $R$ and $Q$, to get that
\[f_i(S_i) \ge \frac{\opt'_i}{\beta(\theta - 2 + 1/\gamma)} \ge \frac{(\theta-1)\fmms_i(r, Q)}{\theta(\theta-1)} \ge \frac{\fmms_i}{\theta}\,,\]
where the first inequality follows from expressing the approximation guarantees of Theorem \ref{thm:greedy_monotone_general_robust} in terms of $\theta$.
\end{proof}

The counterpart of Theorem \ref{thm:greedy_monotone_general_MMS-oracle} for non-monotone objective below is, to the best of our knowledge, the first result on maximin share guarantees for non-monotone objectives. For the sake of presentation, we do not add separate results for cardinality constraints, although one could improve the approximation factor in that case to $7$ rather than $11$.

\begin{theorem}\label{thm:greedy_non-monotone_general_MMS-oracle}
Any agent $i$ with a non-monotone objective and a $p_i$-system constraint, who has access to an $\fmms$ oracle and 
chooses according to the augmented greedy policy $\mathcal{H}_i$ in the Augmented Round-Robin protocol, 
builds two solutions $S_{i1}, S_{i2}$ such that $\max_{t\in\{1, 2\}} f_i(S_{it}) \ge {\fmms_i}/{(4p_i + 7)}$.
\end{theorem}

\begin{proof}
The proof is very similar to the proof of Theorem \ref{thm:greedy_monotone_general_MMS-oracle}, so here we focus mostly  on the differences.
If agent $i$ gets an item $j$ in Phase $1$, then $S_{i1} = S_i = \{j\}$ and $S_{i2} = \emptyset$.
In this case, by lines \ref{line:mms_single}-\ref{line:mms_j} of policy $\mathcal{H}_i$, we directly have 
$f_i(S_i) \ge \fmms_i / \theta$.

Next suppose that agent $i$ participates in Protocol \ref{alg:MRR} that runs on the instance induced by $R$ and $Q$ in Phase $2$. If this happened because $A=\emptyset$ in line \ref{line:mms_single}, then the statement trivially holds for $S_i = \emptyset$. So, assume that this is not the case. Let $\opt'_i:= \max_{S\in \mathcal{I}_i|Q}f_i(S)$ and $r:=|R|$. 
Like in the proof of Theorem \ref{thm:greedy_monotone_general_MMS-oracle}, right before $i$ is allowed to pick their first item, the remaining instance is $\big(r, \frac{\theta\, \opt'_i}{(\theta-1)\fmms_i(r, Q)}\big)$-robust 
with respect to $i$. 

Then we may apply Theorem \ref{thm:greedy_non-monotone_general_robust} with $\gamma = 1$ and $\beta = \frac{\theta\, \opt'_i}{(\theta-1)\fmms_i(r, Q)}$ for agent $i$ on the execution of Protocol \ref{alg:MRR} on the 
instance induced by $R$ and $Q$, to get that
\[\max_{t\in\{1, 2\}} f_i(S_{it}) \ge \frac{\opt'_i}{2\beta(2p_i + 1 + 2/\gamma)} \ge \frac{(\theta-1)\fmms_i(r, Q)}{\theta(\theta-1)} \ge \frac{\fmms_i}{\theta}\,.\qedhere\]
\end{proof}

Theorems \ref{thm:greedy_monotone_general_MMS-oracle} and \ref{thm:greedy_non-monotone_general_MMS-oracle} have an obvious shortcoming. Unlike the unconstrained additive case, where $i$ could use a PTAS to approximate $\mms_i$ \citep{Woeginger97} and, thus, assuming an (almost exact) $\fmms$ oracle is reasonable, there are no known approximation algorithms for computing $\fmms_i$ subject to combinatorial constraints. 
Next, we show that having an oracle is not really necessary; each agent $i$ can efficiently find an estimate of $\fmms_i$  by simulating a number of executions of Protocol \ref{alg:MRR} with copies of themselves who all follow a greedy policy.

\makeatletter
\renewcommand{\ALG@name}{Algorithm}
\makeatother

\setcounter{algorithm}{0}

\begin{algorithm}[ht]
		\caption{ Maximin Share Proxy$(n, M, f_i)$ }
		\begin{algorithmic}[1]
			\vspace{2pt}\State $\hat{\mu}_i = 0$; $\ell = n$; $Q = M$
            \While{$\max_{z\in Q} f_i(z\,|\,\emptyset) > \hat{\mu}_i$}
                \vspace{1pt}\State $j \in  \argmax_{z\in Q} f_i(z\,|\,\emptyset)$
                \State $\ell = \ell - 1$
                \State $Q = Q \setminus \{j\}$
                \If{$f_i$ is monotone}
                    \State Let $\mathcal{S} = (S_1,\ldots,S_{\ell})$ be the allocation produced by running Round-Robin$(\mathcal{G}_i, \mathcal{G}_i, \ldots,\mathcal{G}_i)$  
                    \State $\hat{\mu}_i = \min_{r\in[\ell]} f_i(S_r)$
                \Else
                    \State Let $\mathcal{S} = ((S_{11}, S_{12}),\ldots,(S_{\ell1},S_{\ell2}))$ be the allocation produced by Round-Robin$(\mathcal{G}^+_i, \ldots,\mathcal{G}^+_i)$  
                    \State $\hat{\mu}_i = \min_{r\in[\ell]} \max_{s\in[2]} f_i(S_{rs})$
                \EndIf
            \EndWhile
            \If{$\min_{z\in M\setminus Q} f_i(z\,|\,\emptyset) < \hat{\mu}_i$}
                \vspace{2pt}\State $\hat{\mu}_i = \min_{z\in M\setminus Q} f_i(z\,|\,\emptyset)$
            \EndIf
            \State \Return{$\hat{\mu}_i$}
		\end{algorithmic}
		\label{alg:maximin_proxy}
\end{algorithm}\medskip

\begin{lemma}\label{lem:maximin_proxy}
If an agent $i$ runs Maximin Share Proxy$(n, M, f_i)$, then the output $\hat{\mu}_i$ they get, lies in the interval $[\fmms_i / \theta, \fmms_i]$, where $\theta$ is as in line \ref{line:theta} of Policy \ref{alg:aug_greedy_policy}.
\end{lemma}

\begin{proof}
First we claim that every item ever removed from $Q$ (i.e., the final items of $M\setminus Q$) has value at least $\fmms_i / \theta$. To see this, suppose that this is not the case and call $Q^{(t)}, \ell^{(t)}, j^{(t)}, \hat{\mu}^{(t)}_i$ the relevant quantities at the beginning of the iteration of the \textit{while} loop during which $j^{(t)}$ was removed. If $f_i(j^{(t)}) < \fmms_i / \theta$, then by Theorems \ref{thm:greedy_monotone_general_MMS-oracle} and \ref{thm:greedy_non-monotone_general_MMS-oracle}, we have that the value $\hat{\mu}^{(t)}_i$, computed during the previous iteration must be at least $\fmms_i / \theta$. In this case, however, the condition of the while loop is false, contradicting the fact that the iteration happened and $j^{(t)}$ was removed.

Similarly, at the end of the last iteration of the while loop, either $\hat{\mu}_i$ was computed on an instance without any items of value at least $\fmms_i / \theta$ from $i$'s perspective (and thus, by Theorems \ref{thm:greedy_monotone_general_MMS-oracle} and \ref{thm:greedy_non-monotone_general_MMS-oracle}, $\hat{\mu}_i$ must be at least $\fmms_i / \theta$) or there is still at least an item of value at least $\fmms_i / \theta$  in the final $Q$ (and thus, the fact that the condition of the while loop is false right after, directly implies $\hat{\mu}_i\ge \fmms_i / \theta$). 

In either case, after the while loop, $\hat{\mu}_i$ is updated to be the minimum over a number of values all of which are at least $\fmms_i / \theta$ and the lower bound on the final value of $\hat{\mu}_i$ follows.

In order to see that $\hat{\mu}_i\le \fmms_i$, let $z_1, z_2, \ldots, z_{n-\ell}$ be the items of $M\setminus Q$ and $S_1,\ldots,S_{\ell}$ be the last allocation produced in the while loop (where here $S_j = \argmax_{T\in\{S_{j1}, S_{j2}\}}f_i(T)$ in the case of a non-monotone $f_i$). Then, by the definition of $\hat{\mu}_i$, each one of the $n$ disjoint feasible (for $i$) subsets  $\{z_1\}, \{z_2\}, \ldots, \{z_{n-\ell}\}, S_1, \ldots,S_{\ell}$ has value \emph{at least} $\hat{\mu}_i$. By the definition of the $n$-feasible maximin share, this immediately implies that $\hat{\mu}_i\le \fmms_i$.
\end{proof}

We may modify the augmented greedy policy ${\mathcal{H}}_i$ to use the proxy value $\hat{\mu}_i$ instead of $\fmms_i$; we will call this new augmented greedy policy $\hat{\mathcal{H}}_i$. As $\hat{\mathcal{H}}_i$ will have minimal differences from ${\mathcal{H}}_i$, we do not give its pseudocode description but rather explain these differences: \textit{(i)} instead of $\theta$ in line \ref{line:theta} of Policy \ref{alg:aug_greedy_policy}, in $\hat{\mathcal{H}}_i$ we have 
\[\hat{\theta} = \frac{2\theta - 1}{\theta}  = {\begin{cases}
                \frac{5}{3}\ , & \text{for a monotone objective and a cardinality constraint}  \\
                \frac{2p_i + 5}{p_i + 3}\ , & \text{for a monotone objective and other $p_i$-system constraints}  \\
                \frac{8p_i + 13}{4p_i + 7}\ , & \text{for a non-monotone objective and a $p_i$-system constraint} 
                \end{cases}} \]
and \textit{(ii)} instead of comparing the value of the best available singleton to $\fmms_i / \theta$ in line \ref{line:mms_single} of Policy \ref{alg:aug_greedy_policy}, in $\hat{\mathcal{H}}_i$ we compare it to $\hat{\mu}_i / \hat{\theta}$.

We are finally ready to state the analogs of Theorems \ref{thm:greedy_monotone_general_MMS-oracle} and \ref{thm:greedy_non-monotone_general_MMS-oracle}, without any assumptions about access to $\fmms$ oracles. Note that bypassing the computationally hard problem of agent $i$ computing their $\fmms_i$, comes at the expense of a factor that is roughly $2$.

\begin{theorem}\label{thm:greedy_monotone_general_MMS}
Any agent $i$ with a monotone objective and a $p_i$-system (resp.~cardinality) constraint, who 
chooses according to the augmented greedy policy $\hat{\mathcal{H}}_i$ in the Augmented Round-Robin protocol, 
builds a solution $S_i$ such that $f_i(S_i) \ge {\fmms_i}/{(2p_i + 5)}$ (resp.~$f_i(S_i) \ge {\fmms_i}/{5}$), even without access to an $\fmms$ oracle.
\end{theorem}

\begin{proof}
It suffices to show that, in any case, $f_i(S_i) \ge {\fmms_i}/{(2\theta-1)}$. The proof is similar to the proof of Theorem \ref{thm:greedy_monotone_general_MMS-oracle}, so we mainly focus on the differences.

If agent $i$ gets an item $j$ in Phase $1$, then $S_i = \{j\}$ and  we  have 
\[f_i(S_i) \ge \frac{\hat{\mu}_i}{\hat{\theta}} \ge \frac{\fmms_i}{\theta\,\hat{\theta}} =  \frac{\fmms_i}{2\theta-1}\,,\]
where the second inequality follows by Lemma \ref{lem:maximin_proxy}, and the equality after that by the definition of $\hat{\theta}$.

Next suppose that agent $i$ participates in Protocol \ref{alg:MRR} that runs on the instance induced by $R$ and $Q$ in Phase $2$. If this happened because $A=\emptyset$ in line \ref{line:mms_single}, then the statement trivially holds for $S_i = \emptyset$. Assume that this is not the case and let $\opt'_i:= \max_{S\in \mathcal{I}_i|Q}f_i(S)$ and $r:=|R|$. 
Similarly to the proof of Theorem \ref{thm:greedy_monotone_general_MMS-oracle}, and using the fact that if $f_i(x) <\hat{\mu}_i / \hat{\theta}$ then $f_i(x) < \fmms_i / \hat{\theta}$, we have that  right before $i$ is allowed to pick their first item, the remaining instance is $\big(r, \frac{\hat{\theta}\, \opt'_i}{(\hat{\theta}-1)\fmms_i(r, Q)}\big)$-robust 
with respect to $i$.

Finally, we  apply Theorem \ref{thm:greedy_monotone_general_robust} with $\gamma = 1$ and $\beta = \frac{\hat{\theta}\, \opt'_i}{(\hat{\theta}-1)\fmms_i(r, Q)}$ for agent $i$ on the execution of Protocol \ref{alg:MRR} on the 
instance induced by $R$ and $Q$, to get that
\[f_i(S_i) \ge \frac{\opt'_i}{\beta(\theta - 2 + 1/\gamma)} \ge \frac{(\hat{\theta}-1)\fmms_i(r, Q)}{\hat{\theta}(\theta-1)}  =  \frac{\frac{\theta-1}{\theta}\fmms_i(r, Q)}{\frac{2\theta-1}{\theta}(\theta-1)}\ge \frac{\fmms_i}{2\theta - 1}\,.\qedhere\]
\end{proof}

\begin{theorem}\label{thm:greedy_non-monotone_general_MMS}
Any agent $i$ with a non-monotone objective and a $p_i$-system constraint, who 
chooses according to the augmented greedy policy $\hat{\mathcal{H}}_i$ in the Augmented Round-Robin protocol, 
builds two solutions $S_{i1}, S_{i2}$ such that $\max_{t\in\{1, 2\}} f_i(S_{it}) \ge {\fmms_i}/{(8p_i + 13)}$, even without access to an $\fmms$ oracle.
\end{theorem}

\begin{proof}
It suffices to show that, in any case, $\max_{t\in\{1, 2\}} f_i(S_{it}) \ge {\fmms_i}/{(2\theta-1)}$. The proof is very similar to the proofs of Theorems \ref{thm:greedy_monotone_general_MMS-oracle}, \ref{thm:greedy_non-monotone_general_MMS-oracle} and \ref{thm:greedy_monotone_general_MMS}, so here we focus only on the differences.
If agent $i$ gets an item $j$ in Phase $1$, then $S_{i1} = S_i = \{j\}$ and $S_{i2} = \emptyset$.
In this case, by lines \ref{line:mms_single}-\ref{line:mms_j} of policy $\mathcal{H}_i$, we directly have 
$f_i(S_i) \ge {\hat{\mu}_i}/{\hat{\theta}} \ge {\fmms_i}/{\theta\hat{\theta}} =  {\fmms_i}/{(2\theta-1)}$.

Exactly like in the proof of Theorem \ref{thm:greedy_monotone_general_MMS}, if $A\neq\emptyset$ in line \ref{line:mms_single} and agent $i$ participates in  Protocol \ref{alg:MRR} on the instance induced by $R$ and $Q$ in Phase $2$, then right before $i$ is allowed to pick their first item, the remaining instance is $\big(r, \frac{\hat{\theta}\, \opt'_i}{(\hat{\theta}-1)\fmms_i(r, Q)}\big)$-robust 
with respect to $i$. 

Then we apply Theorem \ref{thm:greedy_non-monotone_general_robust} with $\gamma = 1$ and $\beta = \frac{\hat{\theta}\, \opt'_i}{(\hat{\theta}-1)\fmms_i(r, Q)}$ for agent $i$ on the execution of Protocol \ref{alg:MRR} on the 
instance induced by $R$ and $Q$, to get that
\[\max_{t\in\{1, 2\}} f_i(S_{it}) \ge \frac{\opt'_i}{2\beta(2p_i + 1 + 2/\gamma)} \ge \frac{(\hat{\theta}-1)\fmms_i(r, Q)}{\hat{\theta}(\theta-1)}  \ge \frac{\fmms_i}{2\theta - 1}\,.\qedhere\]
\end{proof}

By simulating a run of Augmented Round-Robin$(\hat{\bm{\mathcal{H}}})$, i.e., our Protocol \ref{alg:AugRR} with the polynomial-time augmented greedy policies for all agents, we directly get corollaries about $\fmms$ guarantees in fair division.

\begin{corollary}\label{cor:fmms_p-system}
For agents with monotone submodular valuation functions and $p$-system constraints, we can efficiently find a $\frac{1}{2p+5}$-FMMS allocation.
\end{corollary}

\begin{corollary}\label{cor:fmms_cardinality}
For agents with monotone submodular valuation functions and cardinality constraints, we can efficiently find a $0.2$-FMMS allocation. 
\end{corollary}

\begin{corollary}\label{cor:fmms_p-system_nm}
For agents with non-monotone submodular valuation functions and $p$-system constraints, we can efficiently find a $\frac{1}{8p+13}$-FMMS allocation.
\end{corollary}

\section{Improving Algorithmic Fairness and Guarantees via Randomness}
\label{sec:randomizedRR}

In this section we rectify an obvious shortcoming of Protocol \ref{alg:MRR}, namely that not all agents are treated equally due to their fixed order. Unfortunately, this inequality issue is inherent to any deterministic protocol which is agnostic to the objective functions (like Protocol \ref{alg:MRR} is) and it heavily affects agents in the presence of a small number of highly valued contested items. A natural remedy, which we apply here, is to randomize over the initial ordering of the agents before running the main part of Protocol \ref{alg:MRR}. This \textit{Randomized Round-Robin} protocol is formally described in Protocol \ref{alg:RMRR} below.

\makeatletter
\renewcommand{\ALG@name}{Protocol}
\makeatother

\setcounter{algorithm}{2}

\begin{algorithm}[ht]
		\caption{Randomized Round-Robin$({\mathcal{A}}_1, \ldots, {\mathcal{A}}_n)$ \\{\small {(For $i\in [n]$, ${\mathcal{A}}_i$ is the policy of   agent $i$.)}}}
		\begin{algorithmic}[1]
            \State Let $\pi:[n]\to[n]$ be a random permutation on $[n]$
			\State $Q=M$\textbf{;} $k = \lceil m/n\rceil$ 
			\For{$r = 1, \dots, k$} 
			\For{$i = 1, \dots, n$} 
			\State $j = \mathcal{A}_{\pi(i)}(S_{\pi(i)}\,; Q)$ {\small\hfill (where $j$ could be a \textit{dummy} item)} \label{line:rrr4}
			
			\State $Q = Q\setminus \{j\}$ \label{line:rrr6}
			\EndFor
			\EndFor
		\end{algorithmic}
		\label{alg:RMRR}
\end{algorithm}

Of course, given a permutation $\pi$, all the guarantees of Theorems \ref{thm:greedy_monotone_general}, \ref{thm:greedy_monotone_cardinality}, \ref{thm:greedy_non-monotone_general}, and \ref{thm:greedy_non-monotone_cardinality} still hold ex-post, albeit properly restated. That is, now the guarantee for an agent $i\in [n]$ is with respect to the value of an optimal solution available to $i$, given that $\pi^{-1}(i)-1$ items have been lost  
to agents $\pi(1), \pi(2), \ldots, \pi(\pi^{-1}(i)-1)$ before $i$ gets to pick their first item.

The next theorem states that, in every worst-case scenario we studied in Sections \ref{sec:positive_monotone} and \ref{sec:non-monotone}, all agents who choose greedily obtain a set of expected value within a constant factor of the best possible worst-case guarantee of $\opt_i / n$ (recall the example from the Introduction).

\begin{theorem}\label{thm:randomized}
Assume agent $i$ chooses greedily in the Randomized Round-Robin protocol. Then, the expected value that  $i$  
obtains (from the best solution they build) is at least ${\opt_i}/{\gamma\,n}$, where 
\begin{itemize}
    \item $\gamma = 2+\frac{p_i}{n}$, if $i$ has a $p_i$-system constraint and a monotone submodular objective;
    \item $\gamma = 2$, if $i$ has a cardinality constraint and a monotone submodular objective;
    \item $\gamma = 5+\frac{4p_i +2}{n}$, if $i$ has a $p_i$-system constraint and a non-monotone submodular objective;
    \item $\gamma = 5+\frac{2}{n}$, if $i$ has a cardinality constraint and a non-monotone submodular objective.
\end{itemize}
\end{theorem}

\begin{proof}
Fix some agent $i\in [n]$ and consider the execution of Protocol \ref{alg:RMRR}. We will need a more refined notation for $\optm_i$, which so far denoted the value of an optimal solution available to agent $i$, given that $i-1$ items have been lost to agents $1$ through $i-1$ before $i$ gets to pick their first item. Let $\optm_i(\pi)$ be the value of an optimal solution available to agent $i$, given that $\pi^{-1}(i)-1$ items have been lost to agents $\pi(1), \pi(2), \ldots, \pi(\pi^{-1}(i)-1)$ before $i$ gets to pick their first item.

Let $j_1, \ldots, j_n$ be $i$'s top $n$ feasible singletons, breaking ties arbitrarily, i.e., $j_1 \in \argmax_{j\in M: \{j\}\in \mathcal{I}_i} f_i(j \,|\, \emptyset)$, and $j_{k+1} \in \argmax_{j\in M\setminus\{j_1, \ldots, j_k\}: \{j\}\in \mathcal{I}_i} f_i(j \,|\, \emptyset)$, for any $k\in [n-1]$. 
Suppose that $\pi^{-1}(i) = \ell$, i.e., $i$ is the $\ell$-th agent to pick an item in each round (which happens with probability $1/n$). 
We make two straightforward observations about the time when $i$ is about to choose an item: \textit{(a)} there is an available feasible item of value at least $f_i(j_{\ell} \,|\, \emptyset)$, and \textit{(b)} the maximum possible lost value from $i$'s perspective is at most as much as the sum of values of items $j_1, \ldots, j_{\ell-1}$, i.e., $\optm_i(\pi) \ge \opt_i -\sum_{k=1}^{\ell-1} f_i(j_{k} \,|\, \emptyset)$.

Let $S$ be the random set (or the best among the two random sets in the non-monotone case) that agent $i$ builds. Of course, the randomness is over the permutation $\pi$. We are going to distinguish two cases, depending on how large $\sum_{k=1}^{n} f_i(j_{k} \,|\, \emptyset)$ is.
Note that this is not a random event and, therefore, we do not need to take conditional expectations.
First, assume that $\sum_{k=1}^{n} f_i(j_{k} \,|\, \emptyset) \ge  \opt_i / \gamma$. Then, according to observation (a) above, we directly have
\[\E[f_i(S)] \ge \sum_{k=1}^{n} \frac{1}{n}f_i(j_{k} \,|\, \emptyset) \ge \frac{\opt_i}{\gamma\, n} \,.\]
Next, assume that $\sum_{k=1}^{n} f_i(j_{k} \,|\, \emptyset) <  \opt_i / \gamma$. Then, for any random permutation $\pi$, according to observation (b) above, we have 
\[\optm_i(\pi) \ge \opt_i - \!\!\sum_{k=1}^{\pi^{-1}(i)-1} \!\! f_i(j_{k} \,|\, \emptyset) \ge \opt_i -\sum_{k=1}^{n} f_i(j_{k} \,|\, \emptyset) \ge \Big( 1 - \frac{1}{\gamma} \Big) \opt_i \,,\]
and thus, we can apply the ex-post guarantees of Theorems \ref{thm:greedy_monotone_general}, \ref{thm:greedy_monotone_cardinality}, \ref{thm:greedy_non-monotone_general}, and \ref{thm:greedy_non-monotone_cardinality}. 

For instance, if $i$ has a $p_i$-system constraint and a monotone submodular objective, Theorem \ref{thm:greedy_monotone_general} now implies that no matter what $\pi$ is, 
\[f_i(S) \ge \frac{1}{(n + p_i)} \Big( 1 - \frac{1}{\gamma} \Big) \opt_i \,.\]
Taking $\gamma = 2+{p_i}/{n}$ in this case, we have 
\[\E[f_i(S)] \ge \frac{\gamma - 1}{\gamma(n + p_i)}  \opt_i  = \frac{1+{p_i}/{n}}{(2+{p_i}/{n})(n + p_i)}\opt_i = \frac{1}{2n+p_i}\opt_i = \frac{\opt_i}{\gamma\, n}\,,\]
thus concluding the first case of the theorem. Verifying that the stated values of $\gamma$ work for the other three cases is just a matter of simple calculations.
\end{proof}

In the second result of this section, we show that the assumptions needed (in terms of robustness) to obtain $O(p_i)$ approximation guarantees with respect to $\opt_i$, are significantly weaker. In particular, Theorem \ref{thm:randomized_robust} provides asymptotically best-possible guarantees for \textit{any} instances that are $(\Omega(n), O(1))$-robust with respect to everyone, no matter what the constants hidden in the asymptotic notation are.

\begin{theorem}\label{thm:randomized_robust}
Assume that an instance is $(\lceil \delta n \rceil , \beta)$-robust with respect to every agent for constant $\beta, \delta \in \mathbb{R}_+$. Also, assume that any agent $i$, who has a $p_i$-system constraint and a (monotone or non-monotone) submodular objective, chooses greedily in the Randomized Round-Robin protocol. Then, in expectation, $i$ achieves a $1/\Theta(p_i)$ fraction of\, $\opt_i$. 
\end{theorem}

\begin{proof}
We are going to prove the statement for the monotone case, using Corollary \ref{cor:greedy_monotone_general_robust_1}. 
For the non-monotone case (using Corollary \ref{cor:greedy_non-monotone_general_robust_1}) there are 
trivial differences which  will be briefly discussed at the end of the proof.

We assume $\delta n \ge 1$; if not, $n$ is constant and the statement follows from Theorem \ref{thm:randomized}. Fix an agent $i\in [n]$ and consider a run of Protocol \ref{alg:RMRR}. 
Suppose that $\pi^{-1}(i) = \ell$, i.e., $i$ is the $\ell$-th agent to pick an item in each round (which happens with probability $1/n$). We consider the case where $\ell \le \lceil \delta n \rceil$. We begin with the easy observation that  the given instance is $((\lceil \delta n \rceil-\ell +1) + \ell -1, \beta)$-robust with respect to  $i$, or equivalently,  $(\gamma n + \ell -1, \beta)$-robust with respect to  $i$, for $\gamma = \gamma(n) = (\lceil \delta n \rceil-\ell +1)/n$. Given that $\pi^{-1}(i) = \ell$, Corollary \ref{cor:greedy_monotone_general_robust_1} guarantees that in this run of Protocol \ref{alg:RMRR} agent $i$ builds a solution of value at least 
\[\frac{\opt_i}{\beta(p_i  + 1 + \frac{1}{\gamma})} = 
\frac{(\lceil \delta n \rceil-\ell +1) \opt_i}{\beta((p_i  + 1)(\lceil \delta n \rceil-\ell +1) + n)} 
\ge \frac{(\lceil \delta n \rceil-\ell +1) \opt_i}{\beta  ((p_i  + 1) (\delta n +1)  + n)}
\ge \frac{(\lceil \delta n \rceil-\ell +1) \opt_i}{\beta n ( 2 \delta (p_i  + 1)  + 1)}\,.\]
Like in the proof of Theorem \ref{thm:randomized}, let $S$ be the random set agent $i$ builds, where the randomness is over the permutation $\pi$. The above observation directly implies
\[\E[f_i(S) \,|\, \pi^{-1}(i) = \ell] \ge \frac{(\lceil \delta n \rceil-\ell +1) \opt_i}{\beta n ( 2 \delta (p_i  + 1)  + 1)}\,.\]
By the law of total expectation, we have
\begin{align*}
 \E[f_i(S)] &= \sum_{\ell=1}^{n}\E[f_i(S) \,|\, \pi^{-1}(i) = \ell] \Pr(\pi^{-1}(i) = \ell )  
            \ge \frac{1}{n} \sum_{\ell=1}^{\lceil \delta n \rceil} \E[f_i(S) \,|\, \pi^{-1}(i) = \ell]  \\ 
            & \ge \frac{1}{n}\, \frac{\opt_i}{\beta n ( 2 \delta (p_i  + 1)  + 1)}  \sum_{\ell=1}^{\lceil \delta n \rceil}  (\lceil \delta n \rceil-\ell +1)
            = \frac{\opt_i}{\beta n^2 ( 2 \delta (p_i  + 1)  + 1)}  \sum_{j=1}^{\lceil \delta n \rceil}  j \\
            & = \frac{\opt_i}{\beta n^2 ( 2 \delta (p_i  + 1)  + 1)} \,  \frac{\lceil \delta n \rceil(\lceil \delta n \rceil +1)}{2} 
            \ge \frac{\delta^2 n^2 \opt_i}{2\beta n^2 ( 2 \delta (p_i  + 1)  + 1)}  \\ 
            & = \frac{\delta^2 }{   4 \beta \delta (p_i  + 1)  + 2\beta}\, \opt_i\,,
\end{align*}
and as $\beta$ and $\delta$ are constant, $( 4 \beta \delta (p_i  + 1)  + 2\beta) / \delta^2 = \Theta(p_i)$, as claimed.

When the objective of agent $i$ is non-monotone, then $S$ denotes the best among the two random sets that agent $i$ builds. The proof is the same with slightly different calculations, due to the guarantee of Corollary \ref{cor:greedy_non-monotone_general_robust_1}. In particular, we would get 
\[\E[f_i(S) \,|\, \pi^{-1}(i) = \ell] \ge \frac{(\lceil \delta n \rceil-\ell +1) \opt_i}{4 \beta n ( \delta (2p_i  + 1)  + 1)}\,,\]
and, applying the law of total expectation,
\[\E[f_i(S)] \ge \frac{\delta^2 }{   4 \beta \delta (2p_i  + 1)  + 4\beta}\, \opt_i\,,\]
that again  is $(4 \beta \delta (2p_i  + 1)  + 4\beta) / \delta^2  = \Theta(p_i)$.
\end{proof}

\section{Discussion and Open Questions}
\label{sec:discussion}
In this work we studied a combinatorial optimization and fair division problem in which $n$ agents, each endowed with a submodular objective function and a combinatorial constraint, seek to maximize their respective objectives simultaneously over a common ground set, with the requirement that all solutions be feasible and disjoint.
We analyzed a family of Round-Robin protocols and showed that agents following simple greedy policies enjoy strong guarantees.
For monotone objectives, an agent $i$ with a $p_i$-system constraint achieves a $1/(n+p_i)$ fraction of the best value available when they first get to choose, improving to a $1/\Theta(p_i)$ fraction on instances that are $(\Omega(n),O(1))$-robust with respect to competition---which is asymptotically best-possible in polynomial time.
For non-monotone objectives, analogous results hold with a constant-factor loss.
Randomizing over the agent ordering translates these ex-post guarantees into ex-ante guarantees with respect to the unconstrained optimum $\opt_i$.
We also establish formal fairness guarantees: greedy policies in our Round-Robin protocol yield approximate FEF1 and FEFu allocations for both monotone and non-monotone objectives.
Moreover, via our Augmented Round-Robin protocol and a self-contained proxy algorithm, we obtain the first constant-factor feasible maximin share (FMMS) guarantees for agents with submodular valuations and general combinatorial constraints, for both monotone and non-monotone objectives---results that are, to the best of our knowledge, entirely new to the literature.

There are several natural directions for future work.

\noindent\textit{Tightening the approximation factors.}
The FEF1 factors we achieve---$1/(p+2)$ for a $p$-system constraint and $0.5$ for a cardinality constraint in the monotone case---and our FMMS factors---$1/(2p+5)$ and $1/5$ respectively---are likely not tight.
Improving these factors or providing matching hardness results, would significantly better our understanding of constrained fair division with submodular valuation functions.

\noindent\textit{Adaptive and richer protocols.}
Our protocol is fully non-adaptive: the order in which agents act is fixed in advance, and the protocol does not use any information about the agents' valuations or constraints.
A natural question is whether mild adaptivity can substantially improve both the optimization and fairness guarantees.
For example, one could imagine a protocol that, after each round, re-orders the remaining agents according to the value of their current solution.
Understanding and characterizing the trade-off between protocol simplicity, transparency, and guarantee quality is an intriguing direction.

\noindent\textit{Richer classes of objective functions.}
Our results rely crucially on submodularity.
Nevertheless, many real-world objectives---influence in networks with complex activation rules, coverage with correlated events, or buyer valuations in combinatorial auctions---are  better modeled by \emph{subadditive} functions. 
Extending our simple protocol framework  to subadditive objectives, even in the single-agent setting with complex constraints, is an interesting challenge that  seems to require new techniques.

\section*{Acknowledgments}
This work was supported by the ERC Advanced 
Grant 788893 AMDROMA ``Algorithmic and Mechanism Design Research in 
Online Markets'', the MIUR PRIN project ALGADIMAR ``Algorithms, Games, 
and Digital Markets'', the NWO Veni project No.~VI.Veni.192.153, and the project MIS 5154714 of the National Recovery and Resilience Plan Greece $2.0$ funded by the European Union under the NextGenerationEU Program.


\bibliographystyle{plainnat}
\bibliography{multi_submod_refs}

\end{document}